\documentclass[12pt]{article}
\usepackage{amsthm}
\usepackage{eurosym}
\usepackage{amsfonts}
\usepackage{amssymb,amsmath}
\usepackage[dvips]{color}
\usepackage{amscd}
\usepackage{tikz}
\usepackage{mathtools}
\usepackage{ytableau}
\usepackage{youngtab}
\usepackage{subcaption}
\usepackage{pgfopts}

\usetikzlibrary{decorations.pathreplacing}

\usetikzlibrary{patterns}

\def\and{\mathrm{and}}

\newtheorem{prop}{Proposition}

\newtheorem{thm}{Theorem}

\newcommand{\be}{\begin{equation}}
\newcommand{\ee}{\end{equation}}
\newcommand{\bea}{\begin{eqnarray}}
\newcommand{\eea}{\end{eqnarray}}
\newcommand{\beas}{\begin{eqnarray*}}
\newcommand{\eeas}{\end{eqnarray*}}
\newcommand{\ba}{\begin{array}}
\newcommand{\ea}{\end{array}}

\newcommand{\w}{\omega}

\newcommand{\nbox}{{\,\lower0.9pt\vbox{\hrule \hbox{\vrule height 0.2 cm \hskip 0.19 cm \vrule height 0.2 cm}\hrule}\,}}

\def\href#1#2{#2}
\begin{document}

\begin{titlepage}
\hfill
\vbox{
    \halign{#\hfil         \cr
           } 
      }  

\hbox to \hsize{{}\hss \vtop{ \hbox{}

}}

%

\vspace*{20mm}
\begin{center}

{\large \textbf{Detecting topology change via correlations and entanglement}}

{\normalsize \vspace{5mm} }

{\large \textbf{from gauge/gravity correspondence} }

{\Large \vspace{ 20mm} }

{\normalsize {Hai Lin${}^{1}$, Keyou Zeng${}^{2}$}  }

{\normalsize \vspace{10mm} }

{\small \emph{${}^1$\textit{Yau Mathematical Sciences Center, Tsinghua University,
Beijing 100084, P. R. China
}} }

{\normalsize \vspace{0.2cm} }

{\small \emph{$^2$\textit{Department of Physics, Tsinghua University,
Beijing 100084, P. R. China
\\
}} }

{\normalsize \vspace{0.4cm} }

\end{center}

\begin{abstract}

We compute a momentum space version of the entanglement spectrum and entanglement entropy of general Young tableau states, and one-point functions on Young tableau states. These physical quantities are used to measure the topology of the dual spacetime geometries in the context of gauge/gravity correspondence. The idea that Young tableau states can be obtained by superposing coherent states is explicitly verified. In this quantum superposition, a topologically distinct geometry is produced by superposing states dual to geometries with a trivial topology. Furthermore we have a refined bound for the overlap between coherent states and the rectangular Young tableau state, by using the techniques of symmetric groups and representations. This bound is exponentially suppressed by the total edge length of the Young tableau. It is also found that the norm squared of the overlaps is bounded above by inverse powers of the exponential of the entanglement entropies. We also compute the overlaps between Young tableau states and other states including squeezed states
and multi-mode entangled states which have similarities with those appeared in quantum information theory.

\end{abstract}

\end{titlepage}

\vskip 1cm

\section{Introduction}

\label{sec_Introduction}\vspace{1pt}\renewcommand{\theequation}{1.%
\arabic{equation}} \setcounter{equation}{0}

The gauge/gravity correspondence \cite%
{Maldacena:1997re,Gubser:1998bc,Witten:1998qj} is a nontrivial duality
between a quantum system without gravity and a quantum theory incorporating
gravity in the bulk. It provides a model for studying quantum gravity by
quantum field theory on the boundary of the spacetime. The duality reveals
the emergence of spacetime geometry from the degrees of freedom on the
asymptotic boundary, and the bulk spacetime emerges dynamically from the
quantum mechanical description that lives in fewer dimensions \cite%
{Rangamani:2016dms,VanRaamsdonk:2010pw,Horowitz:2006ct,Koch:2009gq}. The
most studied example of gauge/gravity correspondence is an exact equivalence
between Type-IIB String Theory on $AdS_{5}\times S^{5}$ and $\mathcal{N}=4$
Supersymmetric Yang-Mills Theory (SYM) on $4$-$d$ Minkowski spacetime. This
correspondence allows us to perform calculations relevant to the string
theory while working in the quantum field theory side. It further provides
us a way to investigate new quantitative features of non-perturbative effect
in quantum gravity, and it greatly enriches our knowledge about
non-perturbative aspects of string theory.

In the context of gauge/gravity correspondence, there are backreacted
geometries that correspond to highly excited states in the field theory
side, such as the bubbling geometries \cite%
{Lin:2004nb,Berenstein:2004kk,Corley:2001zk}. On the gravity side, they have
complicated topologies and geometries, and have interesting features
including backreaction and topology changes \cite{Lin:2004nb,Horowitz:2006ct}%
. On the field theory side, states in the Hilbert space of the quantum field
theory are explicitly mapped to the gravity side by associating the
corresponding droplet configuration to the boundary value at the interior of
the spacetimes \cite{Berenstein:2004kk,Corley:2001zk,Lin:2004nb}.

Study in the field theory side shows that these different configurations
\cite{Lin:2004nb,Berenstein:2004kk,Corley:2001zk} live in the same Hilbert
space. Different geometries correspond to excited states which are treated
on equal footing. Since they live in the same Hilbert space, one can perform
operations that are allowed by quantum mechanics, and can superpose states
and compute transition probabilities between different states, for example
\cite{Berenstein:2017abm,Diaz:2015tda,Brown:2006zk}. Different microstates
can be distinguished from each other, by looking carefully at correlation
functions \cite%
{Skenderis:2007yb,Christodoulou:2016nej,Balasubramanian:2007qv}.

In this paper we focus on states which possess interesting geometric
properties in the gravity side. One class of known examples are composite
states labeled by Young tableaux, also called Young tableau states. Their
geometric properties in the gravity side are known to us \cite%
{Lin:2004nb,Berenstein:2004kk,Corley:2001zk}. Another interesting type of
states are coherent states \cite{Berenstein:2017abm}. The coherent states
are interesting, since the gravity dual of these states have descriptions in
terms of semiclassical geometries \cite%
{Lin:2004nb,Berenstein:2004kk,Corley:2001zk}. There are some interesting
phenomena related to these two kind of states. As pointed out in \cite%
{Berenstein:2017abm} and further verified in our paper that while the
coherent states around vacuum have trivial topology, after superposing them,
one can produce Young tableau states dual to geometries with a distinct
topology. Besides, the number of annuli of the geometries can be
independently predicted from the field theory side \cite{Berenstein:2017abm}%
. In this paper we analyze measuring topology and detecting topology change,
via correlations and entanglement in the context of gauge/gravity
correspondence. Topology change is important and should be taken into
account in a quantum theory of gravity for consistency reasons \cite%
{Hawking:1979zw,Horowitz:1990qb,Dowker:2002hm}. In order to understand the
transition probabilities from the coherent states to Young tableau states,
we need to compute the overlaps between them, and this is what we have done
in this paper. The probabilities calculated from these overlaps describe the
probabilities of topology-changing transition from the geometry with a
trivial topology to the geometry with a distinct non-trivial topology. For
instance, during this transition, the number of the black annuli in the
geometries has changed. These configurations also have similarities with the
fuzzballs \cite{Mathur:2005ai}, which have provided important insights into
the information loss problem.

Quantum entanglement is a generic feature of a many-body quantum system.
Quantum correlations are also an important resource for information
processing, and are important in quantum information theory \cite{Quantum
information}. In our study of the gauge/gravity correspondence, the quantum
field theory side of the duality is an example of a many-body quantum
system. This also indicates that many-body states are very important in the
field theory side. Indeed, there are states that are of interest both in our
setup and in quantum information theory, which will be discussed in our
paper. Among other things, we compute a momentum space version of the
entanglement spectrum and entanglement entropy of Young tableau states.

The theory of the symmetric group naturally arises in the study of many-body
quantum mechanics since the permutation of particles naturally defines an
action of symmetric group on the Hilbert space. Also, there is a useful way
to label many-body identical particles by using symmetric groups. We will
study heavy states labeled by Young tableaux corresponding to the
representations of symmetric groups. Mathematically speaking, the approach
and derivation here use the techniques of symmetric groups and the theories
of representations and characters \cite{Ramgoolam:2008yr,Sagan,James
Kerber,Fulton,Goldschmidt,Corley:2001zk}. From a mathematical point of view,
the study of the symmetric group is intimately related to the theory of
symmetric functions. Therefore, we see that the symmetric function theory is
a powerful tool in our study of quantum many-body physics. On the other
hand, physical intuition may also lead to new insights into mathematics. And
we will present such an example that relates the characters of the symmetric
group to the topological properties of the bubbling geometry.

The organization of this paper is as follows. In Section \ref{sec_Young
tableau states and entanglement}, we describe general Young tableau states
and their entanglement spectrum and entanglement entropy. In Section \ref%
{sec_Coherent states and general Young tableau states}, we compute the inner
product between coherent states and Young tableau states. Afterwards in
Section \ref{sec_Bound of overlap and entanglement entropy}, we analyze the
bound of the overlap between coherent states and Young tableau states. And
then in Section \ref{sec_Generalized superposition formula}, we construct a
generalized expansion formula for general Young tableau states in integral
representation. In Section \ref{sec_Squeezed states, multi-mode entangled
states from Young tableau states}, we discuss squeezed states and multi-mode
entangled states and their overlaps with Young tableau states. Finally, we
discuss our results and draw some conclusions in Section \ref{sec_Discussion}%
.

\section{Young tableau states and entanglement}

\renewcommand{\theequation}{2.\arabic{equation}} \setcounter{equation}{0} %
\renewcommand{\thethm}{2.\arabic{thm}} \setcounter{thm}{0}

\label{sec_Young tableau states and entanglement}

Quantum entanglement is a common feature of a many-body quantum system.
Many-body states are important in the study of the gauge/gravity duality. In
our setup, composite many-body states can be labeled by Young tableaux, in
which the tableaux \cite{Fulton} keep track of appropriate symmetrization of
various indices of many identical particles. Here we consider the composite
states, which are labeled by Young tableaux and sometimes called Young
tableau states. They are not direct product states and we consider their
entanglement entropy.

Let us begin by describing the Hilbert space in our discussion. The Hilbert
space of states have a nature tensor product structure given by the momentum
number $k,$
\begin{equation}
\mathcal{H}=\bigotimes_{k}\mathcal{H}_{k}=\mathcal{H}_{1}\otimes \mathcal{H}%
_{2}\otimes \dots ,
\end{equation}%
where $\mathcal{H}_{k}$ is the Hilbert space for mode $k$. The creation and
annihilation operators for mode $k$ are $a_{k}^{\dagger }$ and $a_{k}$.
Their commutation relations are
\begin{equation}
\lbrack a_{k},~a_{k^{\prime }}^{\dagger }]=k\delta _{kk^{\prime }},
\label{commutation_relation_01}
\end{equation}%
with appropriate normalization convention. The states in mode $k$, with
occupation number $l$, is $t_{k}^{l}=(a_{k}^{\dagger })^{l}\left\vert
0\right\rangle _{k},$where $\left\vert 0\right\rangle _{k}~$is the vacuum of
$\mathcal{H}_{k}$.$~$Constructions in this section work generally for
systems having a similar Hilbert space, but we mention that in the context
of half BPS sector of the SYM in which gauge invariant observables can be
constructed from a single complex matrix $X$, the $t_{k}$ corresponds to $%
\text{Tr}(X^{k})$.

We can also understand this Hilbert space in terms of the conjugacy classes
\cite{Sagan,James Kerber} of the symmetric group. Here $t_{k}$ represents a
conjugacy class of cycle of length $k$. The $\mathcal{H}_{k}$ is spanned by $%
t_{k}^{w_{k}},\;w_{k}=0,1,2\dots $, where $w_{k}$ is the occupation number
for mode $k$. The states $t_{k}^{w_{k}}$ have also been studied in \cite%
{Berenstein:2017abm,Corley:2001zk,
Berenstein:2004kk,Lin:2004nb,Ramgoolam:2008yr,Corley:2002mj,Kristjansen:2002bb,Koch:2008cm,Caputa:2014vaa}%
. A general conjugacy class can be written as $\prod_{k}t_{k}^{w_{k}}$. In
other words, a conjugacy class is uniquely determined and labeled by a
sequence $\vec{w}=(w_{1},w_{2},\dots )$. States in the Hilbert space $%
\mathcal{H}$ can be spanned by the basis
\begin{equation}
\prod_{k}(t_{k})^{w_{k}}:=\prod_{k}\left( a_{k}^{\dagger }\right)
^{w_{k}}\left\vert 0\right\rangle ,
\end{equation}%
where $\left\vert 0\right\rangle =\otimes _{k}\left\vert 0\right\rangle _{k}$
is the vacuum of $\mathcal{H}$. Hence the states $\prod_{k}t_{k}^{w_{k}}$
live in $\mathcal{H}$, with each factor $t_{k}^{w_{k}}$ lives in the subspace%
$~\mathcal{H}_{k}$. The norm is
\begin{equation}
\parallel \prod_{k}t_{k}^{w_{k}}\parallel ^{2}=\prod_{k}k^{w_{k}}w_{k}!,
\end{equation}%
and the inner product $\langle \prod_{k}t_{k}^{w_{k}}\lvert
\prod_{k}t_{k}^{u_{k}}\rangle =0$ if $\vec{w}\neq \vec{u}$. The factor $%
k^{w_{k}}$ is due to the $k$ in the convention of the commutation relation (%
\ref{commutation_relation_01}). In this paper, we sometimes use the symbol $%
\left\vert {t_{k}^{l}}\right\rangle $ to denote the same state ${t_{k}^{l}}$%
, and we also denote $\prod_{k}{t_{k}^{w_{k}}}$ to mean $\otimes _{k}{%
t_{k}^{w_{k}}}$, where the products of states are understood as tensor
products.

For a Young tableau $\lambda $, we write $\lambda \vdash n$ to mean that $%
\lambda $ corresponds to a partition of $n$. From the representation theory
of symmetric group, we know that each Young tableau $\lambda $ is associated
with an irreducible representation of the symmetric group $S_{n}$. We define
a Young tableau state that is associated with the Young tableau $\lambda $ by%
\begin{equation}
\left\vert \lambda \right\rangle =\sum_{\vec{w}\in p(n)}\chi _{\lambda }(%
\vec{w})\prod_{k}\frac{1}{k^{w_{k}}w_{k}!}(t_{k})^{w_{k}},  \label{Def_Young}
\end{equation}%
where $p(n)$ is the set of all partitions of $n$, which is all such $\vec{w}$
that $\sum kw_{k}=w_{1}+2w_{2}+3w_{3}+\dots =n$. Here $\vec{w}$ denotes a
partition and also a conjugacy class. Here $\chi _{\lambda }$ is the
character \cite{James Kerber} of the irreducible representation associated
with $\lambda $, and $\chi _{\lambda }(\vec{w})$ means the value of the
character on the conjugacy class $\vec{w}$, or $\chi _{\lambda
}(\prod_{k}t_{k}^{w_{k}})$. We can view a general Young tableau as a
multipartite system, as it can be expanded by different conjugacy classes of
cycles with various lengths.

These states are dual to bubbling geometries with various droplet
configurations [8$-$12, 14, 31$-$34]. See Figure \ref{Figure_droplet} (a).
The LLM bubbling geometries contain a black and white plane. There is a time
direction and a radial direction perpendicular to this plane. And there is $%
S^{3}\times S^{3}$ fibered over the black and white plane. One $S^{3}$
shrinks smoothly on the black domains and the other $S^{3}$ shrinks smoothly
on the white domains.

\begin{figure}[h]
\centering
\begin{subfigure}{0.27\textwidth}
		\includegraphics[width=\textwidth]{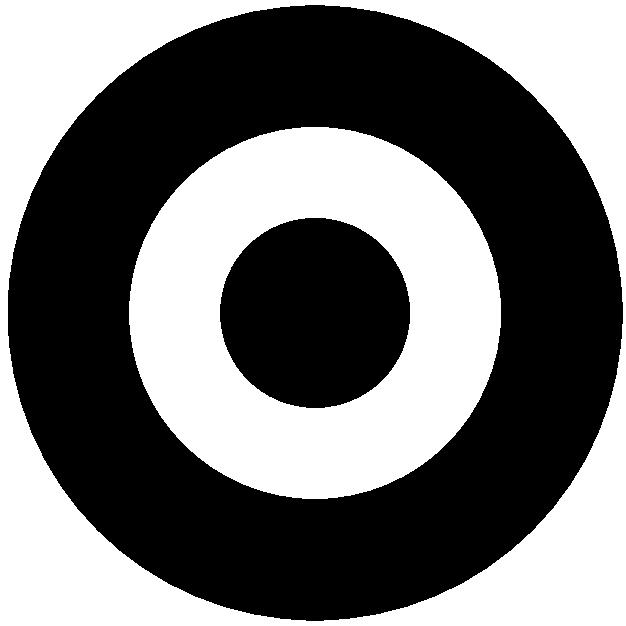}
		\vspace{3pt}
        \caption{State dual to the rectangular Young tableau}
        \label{fig:rectangular01}
	\end{subfigure}
\;\;\;\;\;
\begin{subfigure}{0.3\textwidth}
		\includegraphics[width=\textwidth]{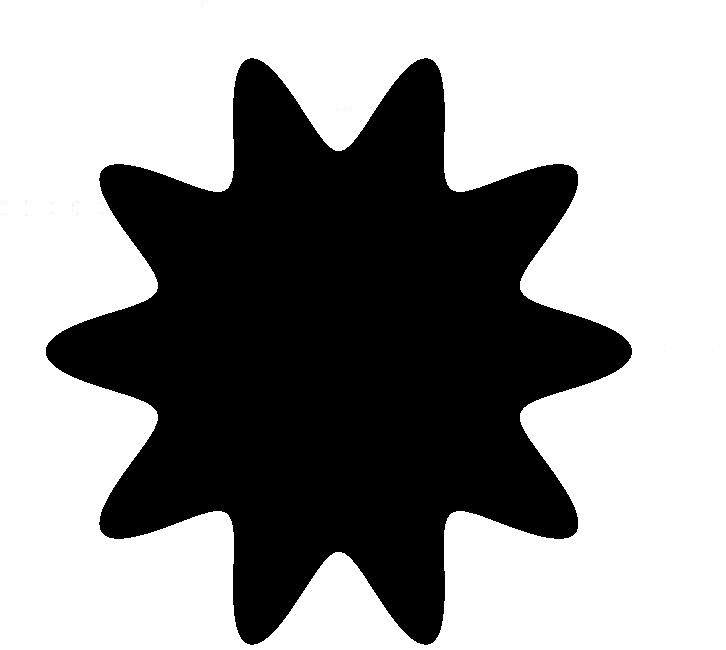}
		\vspace{3pt}
        \caption{State with more symmetry}
        \label{fig:coh_01}
	\end{subfigure}
\;\;\;\;\;
\begin{subfigure}{0.28\textwidth}
		\includegraphics[width=\textwidth]{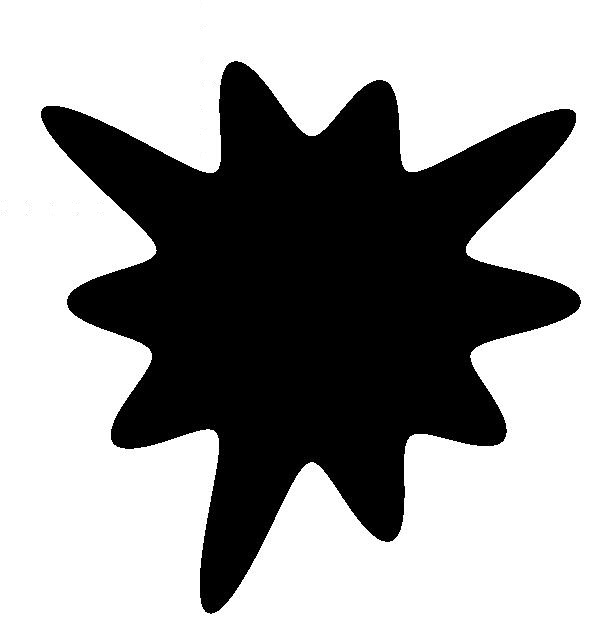}
        \vspace{3pt}
		\caption{State with less symmetry}
		\label{fig:coh_02}
	\end{subfigure}
\caption{This is the black and white plane inside the ten-dimensional LLM
bubbling geometries in string theory. The droplet picture of the geometry
dual to the operator labeled by a Young tableau on the field theory side is
shown in (a). The droplet picture of the geometries dual to the operators
corresponding to the coherent states on the field theory side are shown in
(b, c). The geometries corresponding to states (b) and (c) have the same
topology as the ground state geometry. While, the geometry corresponding to
(a) has a different topology. On the other hand, the states represented by
(b) and (c) participate in the superposition to obtain state (a). }
\label{Figure_droplet}
\end{figure}

Given a Young tableau state $\rvert \lambda \rangle $, it can be mapped to a
droplet configuration, where horizontal edges correspond to white annuli and
vertical edges correspond to black annuli. And the area of the annulus is
determined by the length of the corresponding edge. The operator $\frac{1}{%
\sqrt{j^{l}l!}}(a_{j}^{\dagger })^{l}$ creates excitations whose gravity
dual interpretations are $l$ Kaluza-Klein (KK) gravitons each with momentum $%
j$, moving along the circular direction on the black and white plane in the
bubbling geometries \cite{Lin:2004nb}. We denote the length of each
horizontal edge to be $L_{i}$, and the length of each vertical edge to be $%
M_{i}$. Then the corresponding white and black annuli have areas $L_{i}$ and
$M_{i}$. The corresponding droplets are shown in Figure \ref%
{Figure_Young_general_droplet}.

\begin{figure}[h]
\centering
\vspace{3pt} \includegraphics[width=0.5\textwidth]{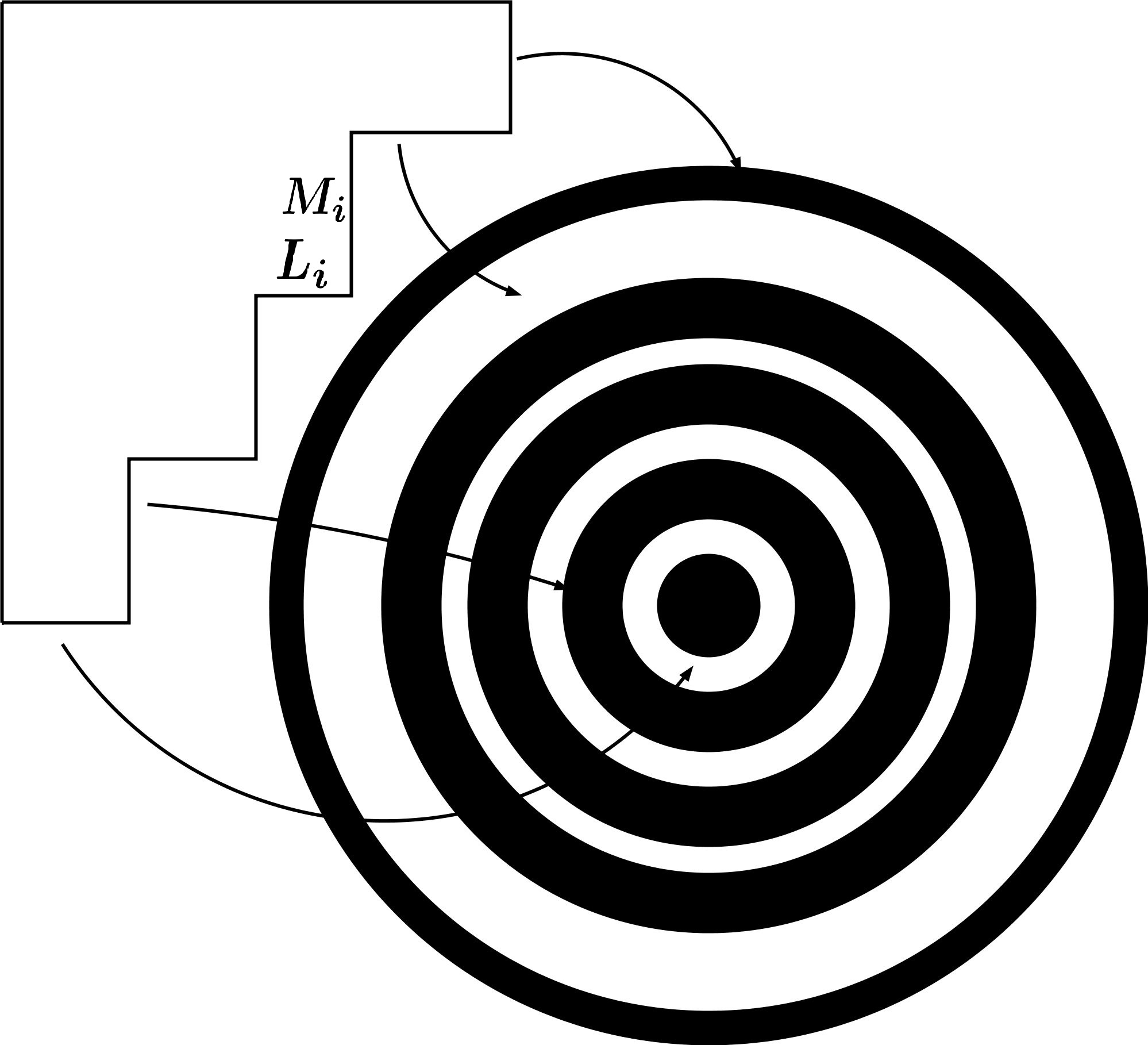}
\vspace{3pt} \vspace{1pt}
\caption{More general Young tableau and the corresponding droplets.}
\label{Figure_Young_general_droplet}
\end{figure}

\vspace{1pt}

There is another type of states, the coherent states \cite%
{Berenstein:2017abm}. A general coherent state can be written as%
\begin{align}
\left\vert Coh\right\rangle & =\prod_{k}\exp (\Lambda _{k}\frac{t_{k}}{k}%
)=\prod_{k}(\sum_{l_{k}=0}^{\infty }\frac{1}{l_{k}!}(\Lambda _{k}\frac{t_{k}%
}{k})^{l_{k}})  \notag \\
& =\sum_{\vec{l}}\prod_{k}\frac{1}{l_{k}!}(\Lambda _{k}\frac{t_{k}}{k}%
)^{l_{k}}  \label{coh_01}
\end{align}%
where $\Lambda _{k}$ are parameters of the coherent states, and the last sum
is over all $\vec{l}=(l_{1},l_{2},\dots )$. Coherent states can also be
mapped to the droplet configuration. They can describe ripples or
deformations with various sizes around vacuum configuration. See Figure \ref%
{Figure_droplet} (b, c) for example.

The Young tableau state $\left\vert \lambda \right\rangle $ is not a direct
product state, and hence it has nonzero entanglement between modes. Their
entanglement entropy can be calculated by explicit partial tracing in the
Hilbert space. Consider the subsystem whose Hilbert space is $\mathcal{H}%
_{j} $, here $j$ parametrizes different modes in the momentum space. The
entanglement entropy for a composite Young tableau state $\left\vert \lambda
\right\rangle $ (where $\lambda \vdash n$) is
\begin{equation}
s_{j}(\lambda )=-\ \mathrm{Tr}\ _{\mathcal{H}_{j}}(\hat{\rho _{j}}\log (\hat{%
\rho _{j}})),  \label{entropy_young_01}
\end{equation}%
where%
\begin{equation}
\hat{\rho _{j}}=\hat{\rho _{j}}(\lambda )=\ \mathrm{Tr}_{j}(\left\vert
\lambda \right\rangle \left\langle \lambda \right\vert ).
\label{density_matrix_young_01}
\end{equation}%
Here $\hat{\rho _{j}}=\hat{\rho _{j}}(\lambda )$ is the density matrix
operator for the subsystem whose Hilbert space is $\mathcal{H}_{j}$. And Eq.
(\ref{entropy_young_01}) is the entanglement entropy in momentum space of
the subsystem whose Hilbert space is $\mathcal{H}_{j}$, and is obtained by
tracing out other subsystems that are complement to it, whose Hilbert space
is $\bigotimes_{k\neq j}\mathcal{H}_{k}$. This is the momentum space version
of the entanglement entropy. It is useful for studying UV/IR entanglement,
where UV and IR modes are entangled. The subsystem in this case, is a region
in the momentum space. The individual Young tableau state is a pure state.
The entanglement entropy $s_{j}$ arises after partial tracing out other
momentum modes living in the momentum space version of the Hilbert space
decomposition. In the calculation of the entropy, we will calculate partial
trace, and we write $\mathrm{Tr}_{k}$ to mean tracing over $%
\bigotimes_{j\neq k}\mathcal{H}_{j}$, or in other words $\mathrm{Tr}_{k}=%
\mathrm{Tr}\ _{\otimes _{j\neq k}\mathcal{H}_{j}}$. This has to be
distinguished from $\mathrm{Tr}\ _{\mathcal{H}_{k}}$ in our notation. These
density matrices and their associated entanglement entropies can be
calculated by using correlation functions or inner products $\langle \lambda
\lvert \prod_{k}{t_{k}^{w_{k}}}\rangle $.

Let us compute%
\begin{align}
\hat{\rho _{j}}(\lambda )& =\ \mathrm{Tr}_{j}(\left\vert \lambda
\right\rangle \left\langle \lambda \right\vert )  \notag \\
& =\ \mathrm{Tr}_{j}(\sum_{\vec{w}\in p(n)}\chi _{\lambda }(\vec{w})\sum_{%
\vec{u}\in p(n)}\chi _{\lambda }(\vec{u})\prod_{k}\frac{1}{k^{w_{k}}w_{k}!}%
\frac{1}{k^{u_{k}}u_{k}!}\left\vert {t_{k}^{w_{k}}}\right\rangle
\left\langle {t_{k}^{u_{k}}}\right\vert ).
\end{align}%
While tracing over $\bigotimes_{k\neq j}\mathcal{H}_{k}$, since $\mathrm{Tr}%
\ _{\mathcal{H}_{k}}(\left\vert {t_{k}^{w_{k}}}\right\rangle \left\langle {%
t_{k}^{u_{k}}}\right\vert )=\langle t_{k}^{u_{k}}\lvert t_{k}^{w_{k}}\rangle
$, so we have
\begin{equation}
\mathrm{Tr}\ _{j}(\prod_{k}\left\vert {t_{k}^{w_{k}}}\right\rangle
\left\langle {t_{k}^{u_{k}}}\right\vert )=\left\vert {t_{j}^{w_{j}}}%
\right\rangle \left\langle {t_{j}^{u_{j}}}\right\vert \prod_{k\neq
j}k^{w_{k}}w_{k}!\delta _{w_{k},u_{k}}~.
\end{equation}%
The vector $\left\vert {t_{j}^{w_{j}}}\right\rangle ~$has to be normalized
in order to calculate the entropy. Define $\left\vert w{_{j}}\right\rangle
_{j}=\frac{1}{\sqrt{j^{w_{j}}w_{j}!}}\left\vert {t_{j}^{w_{j}}}\right\rangle
$, this normalized state has norm 1. Then the density matrix can be written
as%
\begin{equation}
\hat{\rho _{j}}(\lambda )=\sum_{\vec{w}\in p(n)}(\chi _{\lambda }(\vec{w}%
))^{2}(\prod_{k\neq j}(\frac{1}{k^{w_{k}}w_{k}!}))\frac{1}{j^{w_{j}}w_{j}!}%
\left\vert w{_{j}}\right\rangle _{j}\left\langle w{_{j}}\right\vert _{j} ,
\end{equation}%
where $p(n)~$is the set of all partitions of $n$.

We can read out the probability distribution from this expression
\begin{equation}
p_{l}^{(j)}(\lambda )=\sum_{\vec{w}\in p(n),w_{j}=l}(\chi _{\lambda }(\vec{w}%
))^{2}(\prod_{k\neq j}(\frac{1}{k^{w_{k}}w_{k}!}))\frac{1}{j^{l}l!}~.
\label{p^j_l_01}
\end{equation}%
The $p_{l}^{(j)}(\lambda )$ are the non-zero eigenvalues of the density
matrix $\hat{\rho}_{j}(\lambda )$. The $p_{l}^{(j)}(\lambda )$ is also the
probability for mode $j$ to have occupation number $l$, for a given Young
tableau $\lambda $. The number of none-zero eigenvalues of $\hat{\rho _{j}}$
is bounded from above by $n$, which is the total number of boxes of the
Young tableau $\lambda $. \vspace{1pt}The\ Eq. (\ref{p^j_l_01}) is the
entanglement spectrum of a general Young tableau $\lambda $. The $%
p_{l}^{(j)} $ are the corresponding probability distributions, and they
depend on the representation $\lambda $, that is $p_{l}^{(j)}=p_{l}^{(j)}(%
\lambda )$. The density matrices $\hat{\rho}_{i}(\lambda )$ are density
matrices living in the space of Young tableaux.

The characters have the orthogonality relation%
\begin{equation}
\sum_{\vec{w}\in p(n)}\frac{\chi _{\lambda }(\vec{w})\chi _{\mu }(\vec{w})}{%
\prod_{k}k^{w_{k}}w_{k}!}=\delta _{\lambda ,\mu },
\end{equation}%
or specifically
\begin{equation}
\sum_{\vec{w}\in p(n)}\frac{\chi _{\lambda }(\vec{w})^{2}}{%
\prod_{k}k^{w_{k}}w_{k}!}=1.  \label{character_2}
\end{equation}%
We have hence shown that $\sum_{l}p_{l}^{(j)}=1.$

Hence the entropy is%
\begin{equation}
s_{i}(\lambda )=-\sum_{l}p_{l}^{(i)}(\lambda )\log (p_{l}^{(i)}(\lambda ))
\end{equation}%
and this is for a general Young tableau state labeled by $\lambda $.

\begin{thm}
\label{thm_tableaux} For any Young tableau $\lambda $, define the transpose
of $\lambda $, $\lambda ^{T}$ to be the Young tableau whose shape is given
by reflecting the original diagram $\lambda $ along its main diagonal. Then
the states defined by $\lambda $ and $\lambda ^{T}$ have the same
entanglement spectrum, and hence the same entanglement entropy.
\end{thm}

\begin{proof}
	The Specht modules $V^{\lambda}$ of symmetric group have the following property
	\begin{equation}
	V^{\lambda} \cong  V^{(1^n)}\otimes V^{\lambda^T},
	\end{equation}
	where $V^{(1^n)}$ corresponds to the sign representation. This gives us a formula for the character
	\begin{equation}
	\chi_{\lambda}(\vec{w}) = sgn(\vec{w})\chi_{\lambda^T}(\vec{w}).
	\end{equation}
	Since $sgn(\vec{w}) = \pm 1$, this means that the square of the characters for two tableaux conjugate to each other is the same,
	\begin{equation}
	(\chi_{\lambda}(\vec{w}))^2 = (\chi_{\lambda^T}(\vec{w}))^2.
	\end{equation}
	Using the expression for probability distribution
	\begin{equation}
	p^{(j)}_l(\lambda) = \sum_{\substack{w_1,\dots ,w_{j-1},w_{j+1},\dots ,w_n;\\w_1 + \dots +lj + \dots +nw_n= n}}(\chi_{\lambda}(\vec{w}))^2(\prod_{k \neq j }\frac{1}{k^{w_k}w_k!})(\frac{1}{j^{l}l!}),
	\end{equation}
	it follows straightforwardly that
	\begin{equation}
	p^{(j)}_l(\lambda) = p^{(j)}_l(\lambda^T).
	\end{equation}
	This also proves that they have the same entanglement entropy $s_j(\lambda) = s_j(\lambda^T)$.
\end{proof}

\vspace{1pt}

In passing, we mention that the coherent state (\ref{coh_01}) can be written
as
\begin{equation}
\rvert Coh\rangle =\bigotimes_{s=1}^{\infty }\exp (\Lambda _{s}\frac{%
a_{s}^{\dagger }}{s})\rvert 0\rangle _{s}.  \label{coh_03}
\end{equation}%
It's a general fact that a tensor product state has no entanglement between
modes. Therefore the coherent state here, which can be viewed as a momentum
space version of the coherent state, is a special case of this general
result and has zero entropy for each mode. Note\ that there are also other
types of coherent states \cite{Balasubramanian:2007zt} in real space, which
are different from the ones we consider here.

\subsection{Single row and single column states}

\label{sec_Single row and single column states}

Let us now look at simplest Young tableaux, those with a single row or a
single column. Then we will move on to discuss Young tableaux with more
complicated shapes. The state $\left\vert \lambda \right\rangle $ whose
Young tableau is a single row with length $n$, is called a single row state,
denoted as $\left\vert \Delta \right\rangle _{n}$.\ This correspond to the
totally symmetric representation with $\chi _{\lambda }(\vec{w})=1$. From
the expression (\ref{p^j_l_01}) above, and using some combinatoric
techniques, the probability for occupation number $l$ in mode $j$ is
\begin{equation}
p_{l}^{(j)}(n)=\frac{1}{l!j^{l}}\sum_{k=0}^{\lfloor \frac{n-jl}{j}\rfloor }%
\frac{(-1)^{k}}{k!j^{k}}.
\end{equation}%
The entanglement entropy is \vspace{1pt}%
\begin{equation}
s_{j}(n)=\sum_{l=0}^{n}-p_{l}^{(j)}(n)\log p_{l}^{(j)}(n).
\end{equation}%
Using Incomplete Gamma function $\Gamma (s,x):=\int_{x}^{\infty }t^{s-1}e^{-t}%
\mathrm{d}t$, the probability can be written as
\begin{equation}
p_{l}^{(j)}(n)=\frac{\Gamma (1+\lfloor \frac{n-jl}{j}\rfloor ,\frac{-1}{j})}{%
e^{1/j}j^{l}\Gamma (1+l)\Gamma (1+\lfloor \frac{n-jl}{j}\rfloor )}.
\end{equation}

The expression can be simplified in large $n$ limit, and we always assume a
large $n$ limit in the rest of this section. For large $n$, the expression
for the probability distribution can be simplified as%
\begin{equation}
p_{l}^{(j)}=\frac{1}{l!j^{l}}e^{-\frac{1}{j}},  \label{p^j_l_03}
\end{equation}
where we used $p_{l}^{(j)}$ to represent the probability for large $n$. This
is a Poisson distribution with Fisher information being $j$. Consider $%
\lim_{n\rightarrow \infty }s_{j}(n)=s_{j}$,
\begin{equation}
s_{j}=\frac{1}{j}+\sum_{l=0}^{\infty }e^{-1/j}\frac{\log (l!j^{l})}{l!j^{l}}.
\label{s_j_single_row_01}
\end{equation}

Similarly, the state $\left\vert \lambda \right\rangle $ whose Young tableau
is a single column with length $n$, is called a single column state, denoted
as $\left\vert \bigtriangledown \right\rangle _{n}$. This representation has
$\chi _{\lambda }(\vec{w})=sgn(\vec{w})$. This state has the same
entanglement entropy (\ref{s_j_single_row_01}) as $\left\vert \Delta
\right\rangle _{n}$ by the theorem \ref{thm_tableaux}. This situation is a
special case of theorem \ref{thm_tableaux}.

The entanglement entropy $s_{j}$ of single row states, or single column
states, is the von Neumann entropy of Poisson distribution. This has wide
appearances in the theory of information processing, see for example \cite%
{Poisson}.

We then calculate the Renyi entropy. The result can be expressed by
hypergeometric functions. From the probabilities $p_{l}^{(j)}$, the $q$-th
Renyi entropy has the expression
\begin{equation}
s_{j}^{(q)}=\frac{1}{1-q}\log (\sum_{l}(p_{l}^{(j)})^{q}).
\end{equation}%
Inserting Eq. (\ref{s_j_single_row_01}) into the above expression for Renyi
entropy
\begin{equation}
s_{j}^{(q)}=\frac{1}{1-q}\log (\sum_{l=0}^{\infty }\frac{1}{(l!)^{q}}(\frac{1%
}{j^{q}})^{l}e^{-\frac{q}{j}}).
\end{equation}%
In the above expression, mainly we need to evaluate $\sum_{l=0}^{\infty }%
\frac{1}{(l!)^{q}}x^{l}$. Using the definition of hypergeometric functions
\begin{equation*}
F\left[
\begin{matrix}
a_{1} & a_{2}\dots a_{p} \\
b_{1} & b_{2}\dots b_{s}%
\end{matrix}%
;z\right] =\sum_{l=0}^{\infty }\frac{(a_{1})_{l}(a_{2})_{l}\cdots (a_{p})_{l}%
}{(b_{1})_{l}(b_{2})_{l}\cdots (b_{s})_{l}}\frac{z^{l}}{l!}
\end{equation*}%
where we used expression $(a)_{l}=%
\begin{cases}
1 & l=0 \\
a(a+1)\cdots (a+l-1) & l>0%
\end{cases}%
$, then we have that
\begin{equation*}
\sum_{l=0}^{\infty }\frac{1}{(l!)^{q}}x^{l}=F[\underbrace{%
\begin{matrix}
& 0 &  \\
1 & \dots & 1%
\end{matrix}%
}_{q-1};x].
\end{equation*}%
Inserting this back to the expression of Renyi entropy, we get
\begin{equation}
s_{j}^{(q)}=\frac{1}{1-q}\log (F[\underbrace{%
\begin{matrix}
& 0 &  \\
1 & \dots & 1%
\end{matrix}%
}_{q-1};(\frac{1}{j})^{q}]e^{-\frac{q}{j}}),
\end{equation}%
where we have used the hypergeometric functions. For the second Renyi
entropy, we can express the entropy in terms of the modified Bessel function
of the first kind $I_{\alpha }(x)$ (or hypergeometric function $_{2}F_{1}$).
The second Renyi entropy is
\begin{equation}
s_{j}^{(2)}=-\log (I_{0}(2/j)e^{-\frac{2}{j}}).
\end{equation}

Writing $q=1+\epsilon $, then the formula can be written as
\begin{equation}
s_{j}^{(1+\epsilon )}=-\frac{1}{\epsilon }\log (\sum_{l=0}^{\infty }\frac{1}{%
(l!)}(\frac{1}{j})^{l}e^{-\frac{1}{j}}(1+\epsilon \log ({(\frac{1}{l!})}(%
\frac{1}{j})^{l}e^{-\frac{1}{j}})+O(\epsilon ^{2}))).
\end{equation}%
And taking $q\rightarrow 1$ is the same as taking $\epsilon \rightarrow 0$,
\begin{equation}
\lim_{\epsilon \rightarrow 0}s_{j}^{(1+\epsilon )}=\frac{1}{j}+\sum_{l=0}%
\frac{\log (l!j^{l})}{l!j^{l}}e^{-\frac{1}{j}}=s_{j},
\end{equation}%
which is our previous formula Eq. (\ref{s_j_single_row_01}) for von Neumann
entropy.

\subsection{General tableau states}

\label{sec_General tableau states}

Now let us consider more general Young tableaux. For any Young tableau $%
\lambda $,%
\begin{equation}
\langle {\lambda }\rvert (a_{j}^{\dagger }a_{j})^{k}\left\vert \lambda
\right\rangle =\mathrm{Tr}(\sum_{l=0}^{\infty }p_{l}^{(j)}(a_{j}^{\dagger
}a_{j})^{k}\lvert {l}\rangle _{j}\langle {l}\rvert
_{j})=j^{k}\sum_{l=0}^{\infty }l^{k}p_{l}^{(j)}(\lambda ),
\end{equation}%
where $\lvert {l}\rangle _{j}$ is the normalized state of $t_{j}^{l}$, so it
is $\lvert {l}\rangle _{j}=\frac{1}{\sqrt{j^{l}l!}}t_{j}^{l}$, and $%
a_{j}\lvert {l}\rangle _{j}=\frac{jl}{\sqrt{j^{l}l!}}t_{j}^{l-1}$. In the
above we have used that%
\begin{equation}
(a_{j}^{\dagger }a_{j})\lvert {l}\rangle _{j}=jl\frac{1}{\sqrt{j^{l}l!}}%
t_{j}^{l}=jl\lvert {l}\rangle _{j},
\end{equation}%
and
\begin{equation}
(a_{j}^{\dagger }a_{j})^{k}\lvert {l}\rangle _{j}=(jl)^{k}\lvert {l}\rangle
_{j}.
\end{equation}%
On the other hand, the Hamiltonian for mode $j$ is $H_{j}=a_{j}^{\dagger
}a_{j}$, therefore the above formula can be written as
\begin{equation}
\langle {\lambda }\rvert (\frac{H_{j}}{j})^{k}\left\vert \lambda
\right\rangle =\sum_{l=0}^{\infty }l^{k}p_{l}^{(j)}(\lambda ).
\end{equation}%
The above formula can be regarded as calculating moments of the probability
distribution. The higher order moment we know, the more information we know
about the probability distribution. The first order moment is the average
particle number for mode $j$.

Then we can consider the generating function of the above series,
\begin{align}
\langle {\lambda }\rvert \exp (iH_{j}t)\left\vert \lambda \right\rangle &
=\sum_{k=0}^{\infty }\frac{(it)^{k}}{k!}\langle {\lambda }\rvert
(H_{j})^{k}\left\vert \lambda \right\rangle  \notag \\
& =\sum_{l=0}^{\infty }e^{ijt\times l}p_{l}^{(j)}(\lambda ).
\end{align}%
The function on the right hand side is the characteristic function in the
context of probability theory. We denote the generating function $Z_{\lambda
,j}(t)=\langle {\lambda }\rvert \exp (iH_{j}t)\left\vert \lambda
\right\rangle $. Using the explicit expression for Young tableau state in
terms of the character $\chi _{\lambda }$, we can also give an expression
for $Z_{\lambda ,j}(t)$:
\begin{equation}
Z_{\lambda ,j}(t)=\sum_{\vec{w}\in p(n)}(\chi _{\lambda }(\vec{w}%
))^{2}e^{ijw_{j}t}\prod_{k}\frac{1}{k^{w_{k}}w_{k}!}.
\label{generating_character}
\end{equation}%
Then the probability can be calculated to be
\begin{equation}
p_{l}^{(j)}(\lambda )=j\int_{0}^{\frac{2\pi }{j}}\mathrm{d}tZ_{\lambda
,j}(t)e^{-ijlt}.  \label{generating_probability}
\end{equation}%
Combine the above two formulas Eq. (\ref{generating_character}), (\ref%
{generating_probability}), we can derive our previous formula for the
probability distribution Eq. (\ref{p^j_l_01}). This new formula (\ref%
{generating_probability}) provides us an alternative way to calculate
entanglement entropy.

Now we analyze the entanglement entropy of a general tableau state. A
precise formula is not known unless for some simple cases like what we have
discussed in Section \ref{sec_Single row and single column states}. However,
an analysis of it for the tableaux with all the edges long is feasible and
will reveal some connection with the geometric properties of the gravity
dual.

We have that $\hat{N}_{j}=\frac{1}{j}a_{j}^{\dagger }a_{j}~$is the particle
number operator. The $1/j$ factor in the expression of the particle number
operator, is due to the convention of normalization of the creation and
annihilation operators $a_{j}^{\dagger }$ and $a_{j}$ as in (\ref%
{commutation_relation_01}). On the gravity side, this particle number
corresponds to the graviton number, for gravitons with momentum mode $j$.
The operator $\frac{1}{\sqrt{j^{l}l!}}(a_{j}^{\dagger })^{l}$ creates
excitations whose gravity dual interpretations are $l$ KK gravitons each
with momentum $j$, moving along the circular direction on the black and
white plane. The reduced density matrix can also be understood as the
density matrix for this reduced system of gravitons, and the entropy $%
s_{j}(\lambda )$ can be understood as the von Neumann entropy of this
subsystem.

Since we have the expectation value of the particle number operator
\begin{equation}
\langle \hat{N}_{j}\rangle _{\lambda }=\frac{1}{j}\mathrm{Tr}(\hat{\rho}%
_{j}a_{j}^{\dagger }a_{j}),
\end{equation}%
we can see that%
\begin{eqnarray}
\langle \hat{N}_{j}\rangle _{\lambda } &=&\frac{1}{j}\mathrm{Tr}%
(\sum_{l=0}^{\infty }p_{l}^{(j)}a_{j}^{\dagger }a_{j}\lvert {l}\rangle
_{j}\langle {l}\rvert _{j})  \notag \\
&=&\frac{1}{j}\sum_{l=0}^{\infty }p_{l}^{(j)}\langle {l}\rvert
_{j}a_{j}^{\dagger }a_{j}\lvert {l}\rangle _{j}=\sum_{l=0}^{\infty
}lp_{l}^{(j)}(\lambda ).
\end{eqnarray}

We look at every possible probability distribution, with constraints $%
\sum_{l}x_{l}^{(j)}=1$, and $\sum_{l}lx_{l}^{(j)}=\langle \hat{N}_{j}\rangle
_{\lambda }$. Then our probability distribution $p_{l}^{(j)}$ is one of
them. And we find the probability distribution with largest entropy, then
the entropy that correspond to $p_{l}^{(j)}$ must be smaller than the upper
bound. We use a variational method and consider
\begin{equation}
G=-\sum_{l}x_{l}^{(j)}\log (x_{l}^{(j)})+\alpha
(\sum_{l}x_{l}^{(j)}-1)+\beta (\sum_{l}lx_{l}^{(j)}-\langle \hat{N}%
_{j}\rangle _{\lambda }),
\end{equation}%
where $\alpha ,\beta $ are Lagrangian multipliers. Taking derivatives with
respect to $x_{l}^{(j)},\alpha ,\beta $, we get
\begin{align}
-\log (x_{l}^{(j)})-1+\alpha +\beta l& =0,  \notag \\
\sum_{l}x_{l}^{(j)}-1& =0,  \notag \\
\sum_{l}lx_{l}^{(j)}-\langle \hat{N}_{j}\rangle _{\lambda }& =0.
\end{align}%
Solving these equations, and inserting them into the entropy formula $%
s=-\sum_{l}x_{l}^{(j)}\log (x_{l}^{(j)})$, gives exactly the result
\begin{equation}
s_{\max }=(\langle \hat{N}_{j}\rangle _{\lambda }+1)\log (\langle \hat{N}%
_{j}\rangle _{\lambda }+1)-\langle \hat{N}_{j}\rangle _{\lambda }\log
(\langle \hat{N}_{j}\rangle _{\lambda }).  \label{s_j_02}
\end{equation}%
The above formula holds for general $\lambda $ and $j$. The above is the
upper bound for the entropy $s_{j}(\lambda )$, and since thermal
distribution maximizes the entropy, $s_{\max }$ can actually be viewed as
the entropy for the thermal distribution. There is also an analogy to
temperature that we will later provide.

Now we consider Young tableaux with all the edges long, and as special
examples, the rectangular Young tableaux. Recall that, we denote the length
of each horizontal edge to be $L_{i},$ and the length of each vertical edge
to be $M_{i}$. See Figure \ref{Figure_young_general01}. We can define $%
L=\sum_{i}L_{i}$ and $M=\sum_{i}M_{i}$. Incidentally, a rectangular Young
tableau with $L$ rows and $M$ columns is a special example of them.

\begin{figure}[h]
\centering
\vspace{3pt} \includegraphics[width=0.75\textwidth]{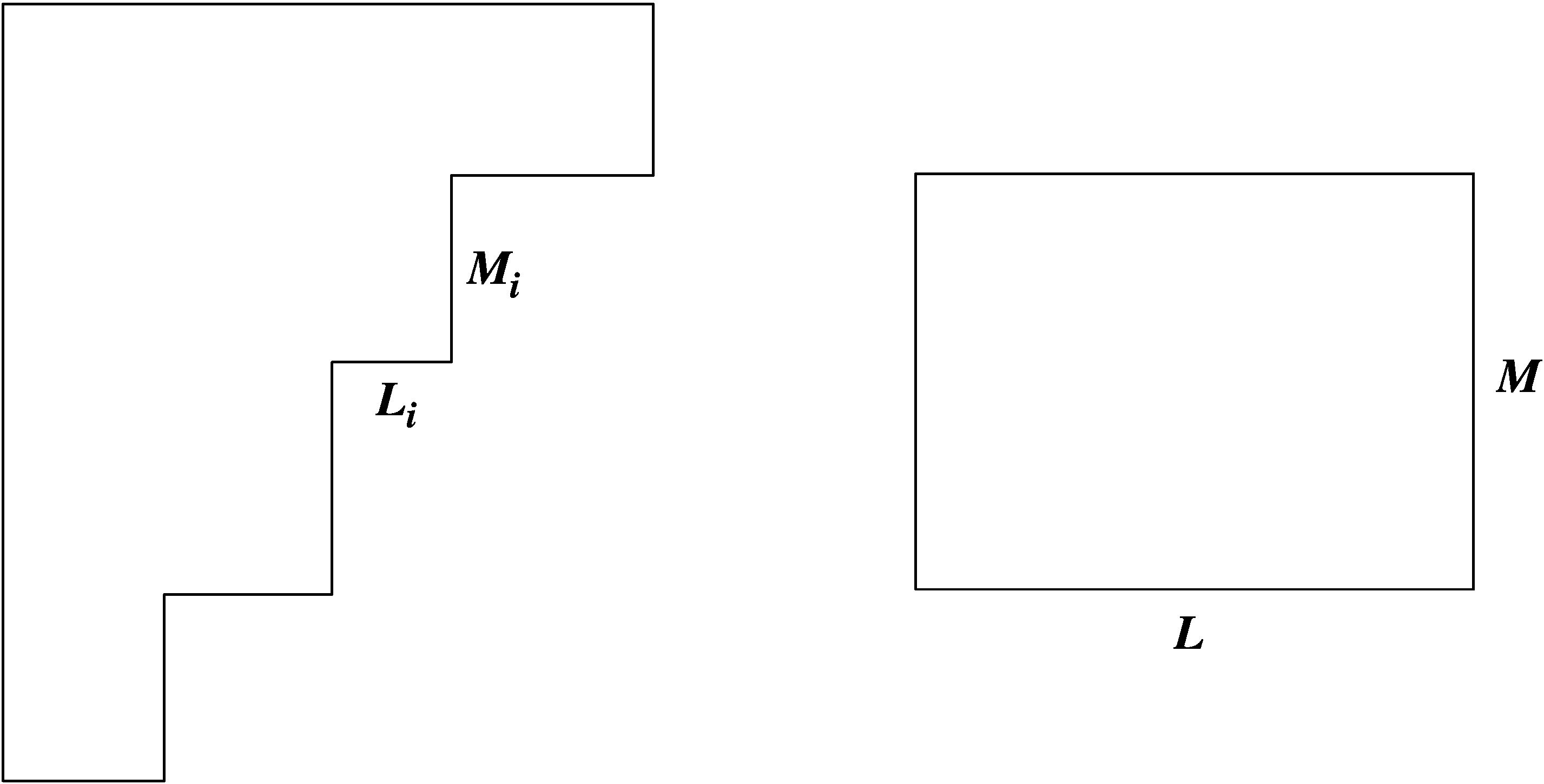}
\vspace{1pt}
\caption{On the left is an example of a general Young tableau with long
vertical edges and long horizontal edges. The $L_i$ and $M_i$ denote the
lengths of each edge. On the right is the rectangular Young tableau with $M$
rows and $L$ columns. The operator labeled by this Young tableau on the
field theory side is dual to the state described in Figure \protect\ref%
{Figure_droplet}(a).}
\label{Figure_young_general01}
\end{figure}

We define a number $n_{\text{cor}}$ to be the number of corners in the lower
right part of the Young tableau. For a Young tableau with all the edges
long, that corresponds to a multi-edge geometry, we have
\begin{equation}
n_{\text{cor}}=n_{\text{anti-edge}},
\end{equation}%
where ${n}_{\mathrm{anti-edge}}$ is the number of the inner edges of the
black annuli in the multi-edge geometry. This is because for such a Young
tableau, it can be mapped to a multi-edge geometry. The number of horizontal
edges is equal to the number of inner edges of black droplets in the
multi-edge geometry, which gives the above formula.

Now we would like to compute quantities like $\langle \hat{N}_{j}\rangle
_{\lambda }$, or more generally $\langle a_{j}^{\dagger k}a_{j}^{k}\rangle
_{\lambda }$. The main technique for the calculation is Murnaghan-Nakayama
rule, which is also presented in \cite{Berenstein:2017abm}. We will
calculate $a_{j}\rvert \lambda \rangle $ or $a_{j}^{k}\rvert \lambda \rangle
$ and then take their norm.

There is a general pattern that when we act a $a_{j}$ on a Young tableau $%
\lambda $, what we get is a sum with coefficient $\pm 1$ of all possible
Young tableau obtained from $\lambda $ by removing, from the lower right
part of $\lambda $, a connected strip-shaped boxes schematically like $%
\ytableausetup{smalltableaux}\ydiagram{2+1,2+1,2+1,3}$ or $%
\ydiagram{5+1,5+1,2+4,2+1,3}$. These are skew Young tableaux.

\vspace{5pt}

We first consider $a_j\rvert \lambda \rangle$ for a multi-edge geometry,
with $j$ small. By $j$ small we mean that%
\begin{equation}
j\leq \min_{i}\{L_{i},M_{i}\},
\end{equation}
where $\min_{i}\{L_{i},M_{i}\}$ denotes the minimum of the set of numbers $%
L_{i},M_{i}$ for all the $i$. In this situation, the possible removing
choice is to remove from every corner a strip of boxes of length $j$ like

\begin{tikzpicture}[xscale=0.9]
\draw   (2,2) -- (2,0.1)--(1,0.1)--(1,0)--(0,0);
\draw [dashed] (1,0)--(2,0)--(2,0.1);
\draw [decorate,decoration={brace,amplitude=5pt,mirror},xshift=0pt,yshift=-3pt]
(1,0) -- (2,0) node [black,midway,yshift=-0.4cm] {\footnotesize $j$};
\node at (2.5,1) {$\pm$};
\draw (5,2) -- (5,0.2)--(4.9,0.2)--(4.9,0.1)--(4.1,0.1)--(4.1,0)--(3,0);
\draw [dashed] (4.1,0)--(5,0)--(5,0.2);
\draw [decorate,decoration={brace,amplitude=5pt,mirror},xshift=0pt,yshift=-3pt]
(4.1,0) -- (5,0) node [black,midway,yshift=-0.4cm] {\footnotesize $j-1$};
\node at (5.5,1) {$\pm$};
\draw  (8,2) -- (8,0.3)--(7.9,0.3)--(7.9,0.1)--(7.2,0.1)--(7.2,0)--(6,0);
\draw [dashed] (7.2,0)--(8,0)--(8,0.3);
\draw [decorate,decoration={brace,amplitude=5pt,mirror},xshift=0pt,yshift=-3pt]
(7.2,0) -- (8,0) node [black,midway,yshift=-0.4cm] {\footnotesize $j-2$};
\node at (8.5,1) {$\pm$};
\node at (9,1) {$\dots$};
\node at (9.5,1) {$\pm$};
\draw (12,2) -- (12,1)--(11.9,1)--(11.9,0)--(10,0);
\draw [dashed] (12,1)--(12,0)--(11.9,0);
\draw [decorate,decoration={brace,amplitude=5pt,mirror},xshift=3pt,yshift=0pt]
(12,0) -- (12,1) node [black,midway,xshift=0.4cm] {\footnotesize $j$};
\end{tikzpicture}

There are signs in the above summation. On the one hand, the sign depend on
the strip shaped boxes that we removed off. The sign is just $(-1)^{h(T)+1}$%
, where $h(T)$ is the height of the strip-shaped boxes. For example $%
\ytableausetup{smalltableaux}\ydiagram{2+1,2+1,2+1,3}$ has $h(T)=4$, and $%
\ydiagram{4+1,4+1,1+4,1+1}$ also has $h(T)=4$. On the other hand, the sign
does not play a role in the final result we want to compute. Because we use
this formula to compute $\frac{1}{j}\langle \lambda \lvert a_{j}^{\dagger
}a_{j}\rvert \lambda \rangle $. And we compute it by computing the norm of $%
a_{j}\rvert \lambda \rangle $. We used Young tableau basis to expand $%
a_{j}\rvert \lambda \rangle =\sum_{\mu }c_{\mu }\rvert \mu \rangle $ as
above. We see in the above formula that $c_{\mu }=1,-1$ or $0$. And since $%
\langle \lambda \lvert a_{j}^{\dagger }a_{j}\rvert \lambda \rangle
=\sum_{\mu }|c_{\mu }|^{2} $, therefore we only need to count the number of
times that $c_{\mu }=\pm 1$.

There are totally $n_{\text{cor}}=n_{\text{anti-edge}}$ such corners and
ever corner contribute $j$ terms, therefore we compute and obtain that%
\begin{equation}
\langle \hat{N}_{j}\rangle _{\lambda }=\frac{1}{j}\mathrm{Tr}(\hat{\rho}%
_{j}a_{j}^{\dagger }a_{j})=\frac{1}{j}\langle {\lambda }\rvert
a_{j}^{\dagger }a_{j}\lvert \lambda \rangle =\frac{1}{j}(jn_{\text{anti-edge}%
})=n_{\text{anti-edge}}~.  \label{N_j_02}
\end{equation}%
When plugging Eq. (\ref{N_j_02}) into Eq. (\ref{s_j_02}), the formula
becomes the same as that of \cite{Berenstein:2017abm}. The derivation here
is alternative to, but in the same spirit as, that in \cite%
{Berenstein:2017abm}, where they used the method of Bogoliubov
transformations. This is also a statement relating the field theory side
with the gravity side. For concentric ring configurations in the bubbling
geometries, the number of black annuli is $n_{\text{\textrm{black-annulus}}%
}=n_{\text{\textrm{anti-edge}}},$ and hence
\begin{equation}
n_{\text{\textrm{black-annulus}}}=\langle \hat{N}_{j}\rangle _{\lambda }.
\label{annulus_01}
\end{equation}%
Since Eq. (\ref{annulus_01}) measures the number of annuli in the geometry,
the field theory quantity tells the information of the topology of the
geometry on the gravity side.

We can also use the definition Eq. (\ref{Def_Young}) of Young tableau state
to calculate $\langle \hat{N}_{j}\rangle _{\lambda }$ and express $\langle
\hat{N}_{j}\rangle _{\lambda }$ in terms of the character corresponding to $%
\lambda $,
\begin{equation}
\sum_{\vec{w}\in p(n)}\frac{\chi _{\lambda }(\vec{w})^{2}w_{j}}{%
\prod_{k}k^{w_{k}}w_{k}!}=n_{\text{anti-edge}},  \label{character_03}
\end{equation}%
with $j\leq \min \{L_{i},M_{i}\}$. We can also derive this equation from the
expression of $p_{l}^{(j)}(\lambda )$ and Eq. (\ref{N_j_02}). The above
formula can be regarded as an extension of the orthogonality relations of
character (Eq. (\ref{character_2})). This formula (\ref{character_03}) also
relates the characters of the representations of the symmetric groups \cite%
{Sagan,James Kerber} to the geometric properties of the bubbling geometries.

The above method can be used to get more relations for the characters of
symmetric group. Consider evaluating the quantity $\langle a_{j}^{\dagger
l}a_{j}^{l}\rangle _{\lambda }$ through calculating $a_{j}^{l}\rvert \lambda
\rangle $ using similar method as above. And we assume that $%
jl<\min_{i}\{L_{i},M_{i}\}$. In this case, when we remove boxes from the
Young tableau, we independently remove boxes from each corner. This is
granted by the condition that $lj<\min_{i}\{L_{i},M_{i}\}$, therefore we
will not reach the boundary of each corner when removing boxes.

As noted in \cite{Berenstein:2017abm} that, removing boxes from the corner
obeys the same rule as creating boxes from the vacuum whenever we do not
exceed the corner. Therefore, we will have
\begin{equation}
\langle a_{j}^{\dagger l}a_{j}^{l}\rangle _{\lambda }=l!j^{l}n_{\text{%
anti-edge}}^{l}.
\end{equation}

The above result gives us an interesting identity involving the character of
symmetric group:
\begin{equation}
\sum_{\vec{w}\in p(n)}\frac{\chi _{\lambda }(\vec{w})^{2}}{%
\prod_{k}k^{w_{k}}w_{k}!}\binom{w_{j}}{l}=n_{\text{anti-edge}}^{l},
\label{character_04}
\end{equation}%
which is a generalization of Eq. (\ref{character_03}).

We then analyze the case for arbitrary $\lambda $ and $j$. Fix $\lambda $
and let $j$ becomes larger, when removing a strip of boxes of length $j$,
there will be the situation that some are not allowed, due to all possible
allowed shapes of a Young tableau. And there will also be the situation that
a strip of boxes will reach two corners or more. Therefore we see that when $%
j>\min_{i}\{L_{i},M_{i}\}$, as $j$ becomes larger, $\langle \hat{N}%
_{j}\rangle _{\lambda }$ will decrease. Finally it becomes zero when $j$
exceeds the total length of every edge, that is
\begin{equation}
\langle \hat{N}_{j}\rangle _{\lambda }=0,\ \ \ \ \
j>\sum_{i}L_{i}+\sum_{i}M_{i}~.
\end{equation}

Using the above relation (\ref{N_j_02}) between $\langle \hat{N}_{j}\rangle
_{\lambda }$ and $n_{\text{anti-edge}}$, we can relate the entanglement
entropy with these quantities. Specifically, for a geometry with large radii
of curvature, $L_{i},M_{i}$ are of order $N$, and hence they are large at
large $N$, and we have the following%
\begin{equation}
s_{j}=s_{\max }.
\end{equation}

There are interesting results for large tableau. Consider $\lambda $
corresponding to a multi-edge geometry and the size of $\lambda $ very
large, for example $\lambda =\Box _{LM}$ with $L,M$ very large. Note that
the maximal entropy $s_{\text{max}}$ is obtained at $x_{l}^{(j)}=e^{\alpha
-1+\beta l}$. And
\begin{equation}
e^{\alpha -1}=\frac{1}{1+\langle \hat{N}_{j}\rangle _{\lambda }},\quad
e^{\beta }=\frac{\langle \hat{N}_{j}\rangle _{\lambda }}{1+\langle \hat{N}%
_{j}\rangle _{\lambda }},
\end{equation}%
which can be written as $p_{l}^{(j)}=x_{l}^{(j)}=(1-x)x^{l}$, with $x=\frac{%
n_{\text{anti-edge}}}{1+n_{\text{anti-edge}}}$ where we have used $\langle
\hat{N}_{j}\rangle _{\lambda }=n_{\text{anti-edge}}$ from our computation (%
\ref{N_j_02}). Since in this case the subsystem reaches the thermal
distribution, ${\tilde{\beta}}=-\beta $ can be viewed as an analog of
inverse temperature for this reduced subsystem. Then by using Eq. (\ref%
{generating_probability}), it follows that
\begin{equation}
\lim_{\text{size of }\lambda \rightarrow \infty }Z_{\lambda ,j}(t)=\frac{1}{%
1+{n}_{\mathrm{anti-edge}}(1-e^{ijt})}.  \label{Z_lambda_j}
\end{equation}%
Here ${n}_{\mathrm{anti-edge}}$ is the number of the inner edges of the
black annuli in the multi-edge geometry, and is equivalent to the number of
long horizontal edges of the Young tableau (see Figure \ref%
{Figure_Young_general_droplet}).

To summarize, $s_{j}=s_{\max }$, in the large $N$ and large $L_{i},M_{i}~$%
limit. At the same time, the large $N$ limit is the limit in which the
gravity dual has large radii of curvature. Away from that, $s_{j}$ is
smaller than $s_{\max },$ and the deviation is in terms of $1/L_{i}$ and $%
1/M_{i}$ corrections. For a Young tableau corresponding to a multi-edge
geometry, and for $j\leq \min_{i}\{L_{i},M_{i}\}$, $\langle \hat{N}%
_{j}\rangle _{\lambda }=n_{\text{anti-edge}}$, as described above. When $%
j>\min_{i}\{L_{i},M_{i}\}~$and as $j$ increases, $\langle \hat{N}_{j}\rangle
_{\lambda }$ will decrease and vanish when $j$ exceeds the total length of
every edge. And $\langle \hat{N}_{j}\rangle _{\lambda }$ is always bounded
by $n_{\text{cor}}$.

\subsection{Rectangular tableau states}

Let us consider rectangular Young tableaux. We can denote a rectangular
Young tableau with $M$ rows and $L$ columns (see Figure \ref%
{Figure_young_general01}) as ${\Box _{LM}}$. The state corresponding to this
rectangular Young tableau is denoted $\lvert {\Box _{LM}}\rangle $. The
gravity dual of this state has the configuration of a black annulus and a
black disk on the black and white plane of the bubbling geometry (see Figure %
\ref{Figure_droplet} (a)). The main step is to calculate $\langle {\Box _{LM}%
}\rvert a_{j}^{\dagger }a_{j}\lvert {\Box _{LM}}\rangle $. The rectangular
tableau ${\Box _{LM}}$ is a special case of a general Young tableau
discussed in Section \ref{sec_General tableau states}. On the one hand
\begin{equation}
\langle {\Box _{LM}}\rvert a_{j}^{\dagger }a_{j}\lvert {\Box _{LM}}\rangle
=\ \mathrm{Tr}\ (\hat{\rho}_{j}a_{j}^{\dagger }a_{j})=\sum_{l=0}^{\infty
}p_{l}^{(j)}\langle {l}\rvert _{j}a_{j}^{\dagger }a_{j}\lvert {l}\rangle
_{j}=j\sum_{l=0}^{\infty }lp_{l}^{(j)}.
\end{equation}%
The above formula can be generalized to
\begin{equation}
\langle {\Box _{LM}}\rvert (a_{j}^{\dagger }a_{j})^{k}\lvert {\Box _{LM}}%
\rangle =j^{k}\sum_{l=0}^{\infty }l^{k}p_{l}^{(j)}.
\end{equation}

We denote
\begin{equation}
y:=\langle \hat{N}_{j}\rangle _{{_{{\Box _{LM}}}}}=\frac{1}{j}\langle \Box
_{LM}\rvert a_{j}^{\dagger }a_{j}\lvert \Box _{LM}\rangle .
\end{equation}%
As derived in Section \ref{sec_General tableau states}, we have shown that
\begin{equation}
s_{max}=(y+1)\log (y+1)-y\log (y),
\end{equation}%
and $s_{j}=s_{max}$ in the large $L,M$ limit.

We then calculate $y$ directly through calculating $a_{j}\lvert \Box
_{LM}\rangle $. The calculation mainly use Murnaghan-Nakayama rule. In the
following calculation, we assume $L\leq M$. Calculation for $L\geq M$ is
slightly different, but the conclusion that as $j$ increases, $y$ does not
increase also hold, and the corresponding result is given in \cite%
{Berenstein:2017abm}.

First for $j\leq L$.

\begin{tikzpicture}[xscale=0.9]
\draw (-3.5,0) --(-3.5,3) -- (-1.5,3) -- (-1.5,0)--(-3.5,0);
\node[left] at (-3.5,1.5) {$a_j$};
\draw [decorate,decoration={brace,amplitude=10pt,mirror},xshift=4pt,yshift=0pt]
(-1.5,0) -- (-1.5,3) node [black,midway,xshift=0.6cm] {\footnotesize $M$};
\draw [decorate,decoration={brace,amplitude=10pt},xshift=0pt,yshift=3pt]
(-3.5,3) -- (-1.5,3) node [black,midway,yshift=0.6cm] {\footnotesize $L$};
\node at (-0.3,1.5) {$=$};
\draw (0,0) --(0,3) -- (2,3) -- (2,0.1)--(1,0.1)--(1,0)--(0,0);
\draw [dashed] (1,0)--(2,0)--(2,0.1);
\draw [decorate,decoration={brace,amplitude=5pt,mirror},xshift=0pt,yshift=-3pt]
(1,0) -- (2,0) node [black,midway,yshift=-0.4cm] {\footnotesize $j$};
\node at (2.5,1.5) {$-$};
\draw (3,0) --(3,3) -- (5,3) -- (5,0.2)--(4.9,0.2)--(4.9,0.1)--(4.1,0.1)--(4.1,0)--(3,0);
\draw [dashed] (4.1,0)--(5,0)--(5,0.2);
\draw [decorate,decoration={brace,amplitude=5pt,mirror},xshift=0pt,yshift=-3pt]
(4.1,0) -- (5,0) node [black,midway,yshift=-0.4cm] {\footnotesize $j-1$};
\node at (5.5,1.5) {$+$};
\draw (6,0) --(6,3) -- (8,3) -- (8,0.3)--(7.9,0.3)--(7.9,0.1)--(7.2,0.1)--(7.2,0)--(6,0);
\draw [dashed] (7.2,0)--(8,0)--(8,0.3);
\draw [decorate,decoration={brace,amplitude=5pt,mirror},xshift=0pt,yshift=-3pt]
(7.2,0) -- (8,0) node [black,midway,yshift=-0.4cm] {\footnotesize $j-2$};
\node at (8.5,1.5) {$-$};
\node at (9,1.5) {$\dots$};
\node at (9.5,1.5) {$\pm$};
\draw (10,0) --(10,3) -- (12,3) -- (12,1)--(11.9,1)--(11.9,0)--(10,0);
\draw [dashed] (12,1)--(12,0)--(11.9,0);
\draw [decorate,decoration={brace,amplitude=5pt,mirror},xshift=3pt,yshift=0pt]
(12,0) -- (12,1) node [black,midway,xshift=0.4cm] {\footnotesize $j$};
\end{tikzpicture}

Since the right hand side have totally $j$ terms, in which each term is a
Young tableau state and they are orthogonal and each have coefficient $\pm 1$%
, therefore we get
\begin{equation}
y=\frac{1}{j}\langle \Box _{LM}\rvert a_{j}^{\dagger }a_{j}\lvert \Box
_{LM}\rangle =\frac{1}{j}j=1,
\end{equation}%
which is $n_{\text{anti-edge}}$ in this case.

Then consider $L\leq j\leq M$.

\begin{tikzpicture}[xscale=0.9]
\draw (0,0) --(0,3) -- (2,3) -- (2,0)--(0,0);
\node[left] at (0,1.5) {$a_j$};
\node at (2.3,1.5) {$=$};
\node at (2.7,1.5) {$\pm$};
\draw (3,0.1) --(3,3) -- (5,3) -- (5,0.5)--(4.9,0.5)--(4.9,0.1)--(3,0.1);
\draw [dashed] (3,0.1)--(3,0)--(5,0)--(5,0.5);
\draw [decorate,decoration={brace,mirror},xshift=2pt,yshift=0pt]
(5,0) -- (5,0.5) node [black,midway,xshift=0.8cm] {\footnotesize $j-L+1$};
\node at (6,1.5) {$\pm$};
\draw (7,0) --(7,3) -- (9,3) -- (9,0.6)--(8.9,0.6)--(8.9,0.1)--(7.1,0.1)--(7.1,0)--(7,0);
\draw [dashed] (7.1,0)--(9,0)--(9,0.6);
\draw [decorate,decoration={brace,mirror},xshift=2pt,yshift=0pt]
(9,0) -- (9,0.6) node [black,midway,xshift=0.8cm] {\footnotesize $j-L+2$};
\node at (9.5,1.5) {$\pm$};
\node at (10,1.5) {$\dots$};
\node at (10.5,1.5) {$\pm$};
\draw (11,0) --(11,3) -- (13,3) -- (13,1)--(12.9,1)--(12.9,0)--(11,0);
\draw [dashed] (13,1)--(13,0)--(12.9,0);
\draw [decorate,decoration={brace,amplitude=5pt,mirror},xshift=3pt,yshift=0pt]
(13,0) -- (13,1) node [black,midway,xshift=0.4cm] {\footnotesize $j$};
\end{tikzpicture}

The right hand side has totally $L$ terms, therefore
\begin{equation}
y=\frac{1}{j}\langle \Box _{LM}\rvert a_{j}^{\dagger }a_{j}\lvert \Box
_{LM}\rangle =\frac{L}{j}.
\end{equation}%
It decrease as $j$ becomes large.

Then consider $M\leq j < L + M$, the calculation is as

\begin{tikzpicture}[xscale=0.9]
\draw (0,0) --(0,3) -- (2,3) -- (2,0)--(0,0);
\node[left] at (0,1.5) {$a_j$};
\node at (2.3,1.5) {$=$};
\node at (2.7,1.5) {$\pm$};
\draw (3,0) --(3,3) -- (4.9,3) -- (4.9,0.5)--(4.9,0.1)--(4,0.1)--(4,0)--(3,0);
\draw [dashed] (4.9,3)--(5,3)--(5,0)--(4,0);
\draw [decorate,decoration={brace,mirror},xshift=0pt,yshift=-2pt]
(4,0) -- (5,0) node [black,midway,yshift=-0.3cm] {\footnotesize $j-M+1$};
\node at (5.5,1.5) {$\pm$};
\draw (6,0) --(6,3) -- (8,3) -- (8,2.9)--(7.9,2.9)--(7.9,0.1)--(6.9,0.1)--(6.9,0)--(6,0);
\draw [dashed] (8,2.9)--(8,0)--(6.9,0);
\draw [decorate,decoration={brace,mirror},xshift=0pt,yshift=-2pt]
(6.9,0) -- (8,0) node [black,midway,yshift=-0.3cm] {\footnotesize $j-M+2$};
\node at (8.5,1.5) {$\pm$};
\draw (9,0) --(9,3) -- (11,3) -- (11,2.8)--(10.9,2.8)--(10.9,0.1)--(9.8,0.1)--(9.8,0)--(9,0);
\draw [dashed] (11,2.8)--(11,0)--(9.8,0);
\draw [decorate,decoration={brace,mirror},xshift=0pt,yshift=-2pt]
(9.8,0) -- (11,0) node [black,midway,yshift=-0.3cm] {\footnotesize $j-M+3$};
\node at (11.5,1.5) {$\dots$};
\node at (12,1.5) {$\pm$};
;\draw (12.5,0) --(12.5,3) -- (14.5,3) -- (14.5,2)--(14.4,2)--(14.4,0.1)--(12.5,0.1);
\draw [dashed] (12.5,0)--(14.5,0)--(14.5,2);
\draw [decorate,decoration={brace,mirror},xshift=2pt,yshift=0pt]
(14.5,0) -- (14.5,2) node [black,midway,xshift=0.8cm] {\footnotesize $j-L+1$};
\end{tikzpicture}

The right hand side has totally $L+M-j$ terms, therefore
\begin{equation}
y=\frac{1}{j}\langle \Box _{LM}\rvert a_{j}^{\dagger }a_{j}\lvert \Box
_{LM}\rangle =\frac{L+M}{j}-1.
\end{equation}

For $j\geq L+M$, then $a_{j}\lvert \Box _{LM}\rangle =0$, therefore $y=0$.

To summarize, we have:
\begin{equation}
y=%
\begin{cases}
1 & j\leq L \\
\frac{L}{j} & L\leq j\leq M \\
\frac{L+M}{j}-1 & M\leq j\leq L+M \\
0 & j\geq L+M%
\end{cases}%
.
\end{equation}%
This means that as $j$ increases, $y$ becomes smaller and decreases from $1$
to $0$. It is interesting that when $j$ exceeds $L+M$, the observable
quantity $\langle \hat{N}_{j}\rangle _{\Box _{LM}}$ becomes zero. For the
rectangular tableau states, $n_{\mathrm{anti-edge}}=1$. This result is also
true if we replace $\Box _{LM}$ by a Young tableau corresponding to a
multi-edge geometry. The proof is almost the same but the discussion will
become more complicated because there will be more situations to discuss.
For the rectangular tableau states, since $n_{\mathrm{anti-edge}}=1$, the
above formula gives $s_{j}=2\log 2$, at large $L,M$, for $j\leq \min \{L,M\}$%
.

\section{Coherent states and general Young tableau states}

\renewcommand{\theequation}{3.\arabic{equation}} \setcounter{equation}{0} %
\renewcommand{\theprop}{3.\arabic{prop}} \setcounter{prop}{0}

\label{sec_Coherent states and general Young tableau states}

As mentioned before, another interesting type of states are coherent states.
A general coherent state can be written as
\begin{align}
\left\vert Coh\right\rangle & =\prod_{k=1}^{\infty }\exp (\Lambda _{k}\frac{%
t_{k}}{k})=\prod_{k=1}^{\infty }(\sum_{l_{k}=0}^{\infty }\frac{1}{l_{k}!}%
(\Lambda _{k}\frac{t_{k}}{k})^{l_{k}})  \notag \\
& =\sum_{\vec{l}}\prod_{k=1}^{\infty }\frac{1}{l_{k}!}(\Lambda _{k}\frac{%
t_{k}}{k})^{l_{k}}
\end{align}%
where the last sum is over all $\vec{l}=(l_{1},l_{2},\dots ,)$. By using the
commutation relations, we see that%
\begin{equation}
a_{k}\left\vert Coh\right\rangle =\Lambda _{k}\left\vert Coh\right\rangle .
\end{equation}%
Hence $\Lambda _{k}$ are the eigenvalues of the $a_{k}$ operators, and by
definition it is a coherent state. Coherent states are very important in
quantum optics and quantum information theory \cite{Quantum information}. A
detailed review of coherent states and their physical and mathematical
implications is in, for example \cite{Zhang:1990fy}. The setup here provides
a new perspective for studying them, namely the gravity dual of coherent
states.

Generic coherent states above the vacuum correspond to the geometries with
the same topology as the vacuum geometry, on the gravity side. These
coherent states, in the dual description on the gravity side, correspond to
creating ripples or deformations \cite%
{Grant:2005qc,Mandal:2005wv,Skenderis:2007yb,Takayama:2005yq} on the vacuum
geometry, and thus do not change the topology of the vacuum geometry. See
Figure \ref{Figure_droplet} (b,c). While, for generic Young tableau states
with long edges, they correspond to geometries with a different topology
than the vacuum geometry. There are various non-contractible cycles that
were absent in the vacuum geometry. See Figure \ref{Figure_droplet} (a).
These different states can be distinguished from each other, by observing
carefully correlation functions \cite%
{Skenderis:2007yb,Christodoulou:2016nej,Balasubramanian:2007qv}. The
coherent states we discuss in this paper are the coherent states around the
vacuum, see also \cite{Berenstein:2017abm,Vazquez:2006id}. These coherent
states have included those corresponding to perturbations around the AdS
vacuum. Here, the geometries dual to coherent states are constructed as
ten-dimensional geometries asymptotic to $AdS_{5}\times S^{5}$ in string
theory. Some classes of these geometries can also be reduced to lower
dimensions and viewed as geometries in lower dimensional gravity \cite%
{Lin:2004nb,Chong:2004ce,Chen:2007du,Liu:2007xj}. Geometries in lower
dimensional gravity that are dual to coherent states have also been
considered in \cite{Gentle:2013fma}. Incidentally, there are other coherent
states as excitations around the Young tableau states \cite%
{Berenstein:2017abm}.

In Ref. \cite{Berenstein:2017abm} they considered coherent states
\begin{equation}
B_{+,\Lambda }\lvert {0}\rangle =\exp (\sum_{k}\Lambda ^{k}\frac{%
a_{k}^{\dagger }}{k})\lvert {0}\rangle .  \label{coh_B_+_02}
\end{equation}%
This corresponds to the case $\Lambda _{k}=\Lambda ^{k}$. Then we can define
coherent states $\lvert {Coh}\rangle =\prod_{k=1}^{\infty }\exp (\Lambda _{k}%
\frac{t_{k}}{k})$, with
\begin{equation}
\Lambda _{k}=\sum_{i}x_{i}^{k},
\end{equation}%
and this can be written as%
\begin{equation}
\left\vert Coh\right\rangle =\left\vert Coh(x_{1},x_{2},\dots )\right\rangle
=B_{+}(x_{1},x_{2},...)\lvert {0}\rangle :=\prod_{i}B_{+,x_{i}}\lvert {0}%
\rangle .
\end{equation}

The dual field theory side is a quantum mechanical system. We can superpose
states and compute transition probabilities between different states. Under
the state-operator correspondence, the inner products or overlaps between
the states are equivalent to correlation functions between the operators
corresponding to the states.

We consider the inner product of coherent states and general tableau states,
and we have the following

\begin{prop}
\label{proposition_coh_01}Consider the coherent state given above
\begin{equation}
\lvert Coh\rangle =\prod_{k=1}^{\infty }\exp (\Lambda _{k}\frac{t_{k}}{k}),
\label{prop_coh_01}
\end{equation}%
in which $\Lambda _{k}=\sum_{i=1}^{m}x_{i}^{k}$. Then $\langle \lambda
|Coh\rangle =s_{\lambda }$ where $s_{\lambda }=s_{\lambda
}(x_{1},x_{2},\dots ,x_{m})$ is the Schur polynomial corresponding to $%
\lambda $.
\end{prop}

\begin{proof}
	Before the proof, we want to discuss the number of variables $x_i$. In fact, we can define Schur polynomial of infinitely many variables \cite{Sagan}, and every Schur polynomial of finitely many variables can be given by the Schur polynomial of infinitely many variables by just sitting extra variables to zero. So our proof works for both finite number of variables and also infinite number of variables.
	
	The proof is straightforward calculation.
	\begin{align}
	\langle \lambda|Coh\rangle &= \sum_{\vec{w} \in p(n)}\chi_{\lambda}(\vec{w})\prod_{k}\frac{1}{k^{w_k}w_k!}\langle t_k^{w_k}\lvert\sum_{\vec{l}} \prod_{k}\frac{1}{l_k!}(\Lambda_k\frac{1}{k})^{l_k}\left\vert {t_k^{w_k}}\right\rangle \nonumber\\
	&= \sum_{\vec{w} \in p(n)}\chi_{\lambda}(\vec{w})\prod_{k}\frac{1}{k^{w_k}w_k!}\frac{1}{w_k!}(\Lambda_k\frac{1}{k})^{w_k}k^{w_k}w_k! \nonumber\\
	& = \sum_{\vec{w} \in p(n)}\chi_{\lambda}(\vec{w})\prod_{k}\frac{1}{k^{w_k}w_k!}(\Lambda_k)^{w_k}.
	\end{align}
	Using formula in \cite{Sagan} (equation (4.23)), when $\Lambda_k = \sum_{i}x_i^k$, the above formula becomes
	\begin{align}
	\langle \lambda|Coh\rangle = s_{\lambda}(x_1,x_2, \dots, x_m).
	\end{align}
\end{proof}

The above proof works in the case of large $N$ limit. In the case of finite $%
N$ the summation of partitions above should be limited to have no more than $%
N$ parts. The proposition \ref{proposition_coh_01} can be easily generalized
to other states. Note that in the definition Eq. (\ref{Def_Young}) of Young
tableau states, the character $\chi _{\lambda }$ play an important role. We
can replace the character $\chi _{\lambda }$ to arbitrary class function
\cite{Sagan,James Kerber} on the symmetric group and this defines another
state. Thus, for a class function $\chi $ on the symmetric group, we define
a state
\begin{equation}
\left\vert \chi \right\rangle =\sum_{\vec{w}\in p(n)}\overline{\chi (\vec{w})%
}\prod_{k}\frac{1}{k^{w_{k}}w_{k}!}(t_{k})^{w_{k}},
\end{equation}%
where $\chi (\vec{w})$ is a general class function of $\vec{w}~$on the
symmetric group, which is in general complex and we put a complex
conjugation on it in our above definition. For $\chi =\chi _{\lambda }$ this
gives back to the definition of Young tableau state $\left\vert \chi
_{\lambda }\right\rangle =\left\vert \lambda \right\rangle $.

We then have the following generalization of Proposition \ref%
{proposition_coh_01}

\begin{prop}
\label{proposition_coh_02} Consider the coherent state $|Coh\rangle$ as in
Proposition \ref{proposition_coh_01} and the state $\left\vert \chi
\right\rangle$ defined by a general class function above. Then $\langle \chi
|Coh\rangle =F(\chi )$, where $F$ is the Frobenius characteristic map which
sends a class function $\chi $ to the corresponding symmetric function $%
F(\chi )$.
\end{prop}

\begin{proof}

	The proof is almost the same as that in Proposition \ref{proposition_coh_01}.
	\begin{align}
	\langle\chi|Coh\rangle &= \sum_{\vec{w} \in p(n)}\chi(\vec{w})\prod_{k}\frac{1}{k^{w_k}w_k!}\langle t_k^{w_k}\lvert\sum_{\vec{l}} \prod_{k}\frac{1}{l_k!}(\Lambda_k\frac{1}{k})^{l_k}\left\vert {t_k^{w_k}}\right\rangle \nonumber\\
	&= \sum_{\vec{w} \in p(n)}\chi(\vec{w})\prod_{k}\frac{1}{k^{w_k}w_k!}\frac{1}{w_k!}(\Lambda_k\frac{1}{k})^{w_k}k^{w_k}w_k! \nonumber\\
	& = \sum_{\vec{w} \in p(n)}\chi(\vec{w})\prod_{k}\frac{1}{k^{w_k}w_k!}(\Lambda_k)^{w_k}.
	\end{align}
	Comparing with the definition of the Frobenius characteristic map (See \cite{Sagan} Definition 4.23), when $\Lambda_k = \sum_{i}x_i^k$, the above formula becomes
	\begin{align}
	\langle \chi|Coh\rangle = F(\chi).
	\end{align}
\end{proof}

As a special case, the Frobenius characteristic map sends the character $%
\chi _{\lambda }$ to the corresponding Schur function. Then we have $F(\chi
_{\lambda })=s_{\lambda }$, which coincides with our previous Proposition %
\ref{proposition_coh_01}.

To summarize, our above result can be stated as%
\begin{equation}
\langle \chi \rvert \prod_{i}B_{+,x_{i}}\lvert {0}\rangle =F(\chi
)(x_{1},x_{2},\dots ).
\end{equation}%
Or specifically for Young tableau states%
\begin{equation}
\langle \lambda \rvert \prod_{i}B_{+,x_{i}}\lvert {0}\rangle =s_{\lambda
}(x_{1},x_{2},\dots ).
\end{equation}

The Young tableau states provide an orthonormal basis and a coherent state
can be expanded by the Young tableau states. And our above formula gives us
the expansion coefficient for the superposition:
\begin{equation}
\prod_{i}B_{+,x_{i}}\lvert {0}\rangle =\sum_{\lambda }s_{\lambda
}(x_{1},x_{2},\dots )\lvert {\lambda }\rangle ,
\end{equation}%
where the summation is over all possible Young tableaux $\lambda $. This
equation shows that the coherent states can be written as the quantum
superpositions of Young tableau states.

Now we look at some special cases. For example, if we only have one variable
$x_{1}$, then the above formula is just the expansion (4.4) in \cite%
{Berenstein:2017abm}, and for two variables $x_{1},x_{2}$, the above formula
gives Eq. (4.43) in \cite{Berenstein:2017abm}.

In Ref. \cite{Berenstein:2017abm} another dual version of coherent states is
defined%
\begin{equation}
B_{-,\Lambda }\lvert {0}\rangle =\exp (-\sum_{k}\Lambda ^{k}\frac{%
a_{k}^{\dagger }}{k})\lvert {0}\rangle .  \label{coh_B_-_02}
\end{equation}%
Note that in the definition, there is an important minus sign in state (\ref%
{coh_B_-_02}), with respect to state (\ref{coh_B_+_02}). Operators $B_{\pm
,\Lambda }$ will be very useful. For example, for the value $\Lambda
=e^{i\gamma }$, the corresponding picture of $B_{+,\Lambda }\lvert {0}%
\rangle $ is a Dirac delta function centered at angle $\gamma $.

There is duality between the $B_{+}$ and $B_{-}$, which is connected to the
duality between Young tableau state $\lvert {\lambda }\rangle $ and its
transpose $\lvert {\lambda ^{T}}\rangle $. We have the following results.

\begin{prop}
\label{proposition_coh_03}We have a duality of inner product between a Young
tableau state and a coherent state as follows
\begin{equation}
\langle \lambda \rvert \prod_{i}B_{+,x_{i}}\lvert {0}\rangle =\langle {%
\lambda ^{T}}\rvert \prod_{i}B_{-,-x_{i}}\lvert {0}\rangle .
\label{prop_coh_03}
\end{equation}
\end{prop}

\begin{proof}
	We compute $\langle{\lambda^T}\rvert\prod_{i}B_{-,-x_i}\lvert{0}\rangle$ to show that it is equal to $s_{\lambda}$.
	\begin{align}
	\langle{\lambda^T}\rvert\prod_{i}B_{-,-x_i}\lvert{0}\rangle &= \sum_{\vec{w} \in p(n)}\chi_{\lambda^T}(\vec{w})\prod_{k}\frac{1}{k^{w_k}w_k!}\langle{t_k^{w_k}}\rvert\sum_{\vec{l}} \prod_{k}\frac{-1}{l_k!}(\sum_i (-x_i)^k\frac{1}{k})^{l_k}\lvert{t_k^{l_k}}\rangle \nonumber\\
	& = \sum_{\vec{w} \in p(n)}\chi_{\lambda^T}(\vec{w})\prod_{k}\frac{-1}{k^{w_k}w_k!}({(-1)^k}\Lambda_k)^{w_k}
	\end{align}
	where $\Lambda_k = \sum_i x_i^k$. Now that $\chi_{\lambda^T}(\vec{w}) = sgn(\vec{w})\chi_{\lambda}(\vec{w})$. And sgn is the $\pm$ according to whether $\vec{w}$ correspond to an odd or even permutation. It can be shown that $sgn(\vec{w}) = -(-1)^{w_1 + 2w_2 + 3w_3+\cdots}$. So we insert $\chi_{\lambda^T}(\vec{w}) = -(-1)^{w_1 + 2w_2 + 3w_3+\cdots}\chi_{\lambda}(\vec{w})$ into the above formula and get
	\begin{align}
	\langle{\lambda^T}\rvert\prod_{i}B_{-,-x_i}\lvert{0}\rangle & = \sum_{\vec{w} \in p(n)}-(-1)^{w_1 + 2w_2 + 3w_3+\cdots}\chi_{\lambda}(\vec{w})\prod_{k}\frac{-1}{k^{w_k}w_k!}({(-1)^k}\Lambda_k)^{w_k} \nonumber\\
	& = \sum_{\vec{w} \in p(n)}\chi_{\lambda}(\vec{w})\prod_{k}\frac{1}{k^{w_k}w_k!}(\Lambda_k)^{w_k} \nonumber\\
	& = s_{\lambda}(x_1,x_2,\dots).
	\end{align}
\end{proof}

Hence we also have
\begin{equation}
\prod_{i}B_{-,x_{i}}\lvert {0}\rangle =\sum_{\lambda }s_{\lambda
}(-x_{1},-x_{2},\dots )\lvert {\lambda }^{T}\rangle .
\end{equation}

The norm-squared $\parallel B_{+,x_{1}}B_{+,x_{2}}\lvert {0}\rangle
\parallel ^{2}$ was computed in Eq. (4.44) in Ref. \cite{Berenstein:2017abm}%
. We can generalize their formula to arbitrary many $B_{+,x_{i}}$, and we
have the following result
\begin{equation}
\langle {0}\rvert (\prod_{i}B_{+,y_{i}}^{\dagger
})(\prod_{i}B_{+,x_{i}})\lvert {0}\rangle =\prod_{i,j}\frac{1}{1-x_{i}\bar{y}%
_{j}}.  \label{inner_product_02}
\end{equation}%
The proof of this formula (\ref{inner_product_02}) is in Appendix \ref%
{appendix_inner product}. It is easy to see the norm of the $%
\prod_{i}B_{+,x_{i}}\lvert {0}\rangle $ from Eq. (\ref{inner_product_02}),
when identifying $y_{i}$ with $x_{i}$, hence
\begin{equation}
\parallel \prod_{i}B_{+,x_{i}}\lvert {0}\rangle \parallel ^{2}=\prod_{i,j}%
\frac{1}{1-x_{i}\bar{x}_{j}}.
\end{equation}%
There is a dual version of the above formula (\ref{inner_product_02}) for
operators $B_{-,x_{i}}$, which gives the same result as above
\begin{equation}
\langle {0}\rvert (\prod_{i}B_{-,y_{i}}^{\dagger
})(\prod_{i}B_{-,x_{i}})\lvert {0}\rangle =\prod_{i,j}\frac{1}{1-x_{i}\bar{y}%
_{j}}.  \label{inner_product_03}
\end{equation}%
There is also a formula for inner product of $B_{+}$ with $B_{-}$, which is
as follows
\begin{equation}
\langle {0}\rvert (\prod_{i}B_{-,y_{i}}^{\dagger
})(\prod_{i}B_{+,x_{i}})\lvert {0}\rangle =\prod_{i,j}(1-x_{i}\bar{y}_{j}).
\label{inner_product_04}
\end{equation}%
The proofs of the above formulas (\ref{inner_product_03}) and (\ref%
{inner_product_04}) are given in Appendix \ref{appendix_inner product}.

Now we discuss Schur polynomials for some simple states, such as the single
row states and single column states, and then rectangular tableau states.
For the Young tableau $\lambda = (n)$, which is a row of $n$ boxes, the
Schur polynomial is
\begin{equation}
s_{(n)}(x_{1},x_{2},\dots )=\sum_{i_{1}\leq i_{2}\leq \dots
i_{n}}x_{i_{1}}x_{i_{2}}\cdots x_{i_{n}}.
\end{equation}%
And we will write this function as $h_{n}:=s_{(n)}$.

For the Young tableau $\lambda =(1^{n})$, which is a column of $n$ boxes,
the Schur polynomial is
\begin{equation}
s_{(1^{n})}(x_{1},x_{2},\dots )=\sum_{i_{1}<i_{2}<\cdots
i_{n}}x_{i_{1}}x_{i_{2}}\cdots x_{i_{n}}.
\end{equation}%
And we will write this function as $e_{n}:=s_{(1^{n})}$.

For a general Young tableau\ with $m$ rows, $\lambda =(\lambda _{1},\lambda
_{2},\dots ,\lambda _{m}),$
\begin{equation}
s_{\lambda }=\det (h_{\lambda _{i}-i+j})_{1\leq i,j\leq m}.
\end{equation}%
In the above formula, it is possible that some $\lambda _{i}-i+j$ is zero,
we assume that $h_{n}$ with negative $n$ is zero.

Apply this formula for a rectangular tableau state $\lambda =(\underbrace{
L,L,\dots ,L}_{M})$, the $s_{LM}$ has expression
\begin{equation}
s_{LM}(x_{1},x_{2},\dots )=%
\begin{vmatrix}
h_{L} & h_{L+1} & \cdots & h_{L+M-1} \\
h_{L-1} & h_{L} & \cdots & h_{L+M-2} \\
\vdots &  &  & \vdots \\
h_{L-M+1} & h_{L-M+1} & \cdots & h_{L}%
\end{vmatrix}%
.
\end{equation}

There is another expression for $s_{LM}$ which is useful
\begin{equation}
s_{LM}=\sum_{\{i_{kl}\}\in T_{LM}}\prod_{k,l}x_{i_{kl}}\ ,
\end{equation}%
where the summation is over all possible index $i_{k,l}\in T_{LM}$, which is
described by the following constraint
\begin{equation}
\begin{matrix}
i_{11} & \leq & i_{12} & \leq & \cdots & i_{1L} \\
\wedge & ~ & \wedge & ~ & ~ & \wedge \\
i_{21} & \leq & i_{22} & \leq & \cdots & i_{2L} \\
\wedge & ~ & \wedge & ~ & ~ & \wedge \\
\vdots &  & \vdots &  &  & \vdots \\
i_{M1} & \leq & i_{M2} & \leq & \cdots & i_{ML}%
\end{matrix}%
.
\end{equation}%
These two descriptions of $s_{LM}$ are equivalent.

An immediate consequence of the above formula is
\begin{equation}
s_{LM}(x_{1},x_{2},\dots ,x_{m},0,0,\dots )=0,\;\;\text{for }m<M,
\end{equation}%
which translates to the following formula
\begin{equation}
\langle {\Box _{LM}}\rvert \prod_{i=1}^{m}B_{+,x_{i}}\lvert {0}\rangle
=0,\;\;\text{for }m<M.
\end{equation}%
More generally, for a Young tableau with $M$ rows, we have
\begin{equation}
\langle {\lambda }\rvert \prod_{i=1}^{m}B_{+,x_{i}}\lvert {0}\rangle =0,\;\;%
\text{for }m<M.
\end{equation}

The Schur polynomial for a general Young tableau is complicated, however for
a rectangular Young tableau $\Box_{LM}$ and $m = M$, we have a simple
formula,

\begin{equation}
\langle {\Box _{LM}}\rvert \prod_{i=1}^{M}B_{+,x_{i}}\lvert {0}\rangle
=s_{LM}(x_{1},x_{2},\dots ,x_{M})=(\prod_{i=1}^{M}x_{i})^{L}.
\end{equation}%
The above formula will be crucial in the next section, where we will further
analyze the overlap.

\section{Bound of overlap and entanglement entropy}

\label{sec_Bound of overlap and entanglement entropy}\vspace{1pt}

\renewcommand{\theequation}{4.\arabic{equation}} \setcounter{equation}{0}

In this section we further analyze the overlap $\langle {\lambda }\rvert
\prod_{i=1}^{M}B_{+,x_{i}}\lvert {0}\rangle $ for $\lvert {\lambda }\rangle $
a rectangular tableau state with $M$ rows and $L$ columns, written as $%
\rvert \Box _{LM}\rangle $. The Young tableau states are normalized with $%
\langle \Box _{LM}\lvert \Box _{LM}\rangle =1$, however, the coherent states
are not normalized therefore we must take their norms into account. We
compute the upper bound of $|\langle {\Box _{LM}}\rvert
\prod_{i=1}^{M}B_{+,x_{i}}\lvert {0}\rangle |$ divided by the norm of $%
\prod_{i=1}^{M}B_{+,x_{i}}\lvert {0}\rangle $. Our analysis of the
properties of the overlap reveals interesting physics.

We will analyze the function%
\begin{equation}
f(x_{1},x_{2},\dots ,x_{M}):=\frac{|\langle {\Box _{LM}}\rvert
\prod_{i=1}^{M}B_{+,x_{i}}\lvert {0}\rangle |^{2}}{\Vert
\prod_{i=1}^{M}B_{+,x_{i}}\lvert {0}\rangle \Vert ^{2}},
\end{equation}%
which is the normalized overlap between coherent states and the rectangular
tableau states $\lvert {\Box _{LM}}\rangle ${.} We also denote $\lvert
Coh(x_{1},\dots ,x_{M})\rangle =\prod_{i=1}^{M}B_{+,x_{i}}\lvert {0}\rangle $%
. The normalized overlap is%
\begin{equation}
f(x_{1},x_{2},\dots ,x_{M})=\prod_{i,j}(1-x_{i}\bar{x}_{j})\times
(\prod_{i=1}^{M}|x_{i}|^{2})^{L}.
\end{equation}

In Appendix \ref{appendix_bound of overlap} we make the derivation of the
supremum of the above normalized inner product. And the final result is%
\begin{equation}
\sup_{\{x_{i}\}}\left\vert \frac{\langle {\Box _{LM}}\rvert
\prod_{i=1}^{M}B_{+,x_{i}}\lvert {0}\rangle }{\Vert
\prod_{i=1}^{M}B_{+,x_{i}}\lvert {0}\rangle \Vert }\right\vert ^{2}=(\frac{M%
}{L+M})^{M}(\frac{L}{L+M})^{L}.  \label{supremum_01}
\end{equation}%
\vspace{1pt}This is the supremum of the normalized overlap. This is the most
strict upper bound, in the sense that there is a state that can actually
saturate this upper bound.

When $L$ and $M$ are both large, the areas of the white annulus and black
annulus are both large, and the geometry have large radii of curvature. We
can consider the behavior of this upper bound at both large $L$ and large $M$%
, and the above formula gives rise to%
\begin{equation}
\sup_{\{x_{i}\}}\left\vert \frac{\langle {\Box _{LM}}\rvert
\prod_{i=1}^{M}B_{+,x_{i}}\lvert {0}\rangle }{\Vert
\prod_{i=1}^{M}B_{+,x_{i}}\lvert {0}\rangle \Vert }\right\vert ^{2}\leq
2^{-(L+M)},  \label{bound_02}
\end{equation}%
where the equal sign is taken when $L=M.$ We think that this result is very
simple and beautiful.

We can use the entanglement entropies to quantify the bound of the overlaps,
or equivalently the correlation functions. The bound is related to the
entanglement entropies, and can be written as
\begin{equation}
\left\vert \frac{\langle {\Box _{LM}}\rvert Coh(x_{1},\dots ,x_{M})\rangle }{%
\Vert Coh(x_{1},\dots ,x_{M})\Vert }\right\vert ^{2}\lesssim \exp \left( -%
\frac{1}{2}\sum_{j=1}^{L+M}s_{j}({\Box _{LM})}\right) ,  \label{bound_03}
\end{equation}%
where the sign $\lesssim $ means smaller than or of the same order. This is
a slightly weaker bound than (\ref{bound_02}). The bound here is sharper
than the bound in \cite{Berenstein:2017abm}. It is a refined version of the
bound in \cite{Berenstein:2017abm}. This is in perfect agreement with the
prediction in \cite{Berenstein:2017abm}, for large $L,M$. This overlap
quantifies the fidelity \cite{Uhlmann} of the two states $\lvert {\Box _{LM}}%
\rangle $ and $\lvert Coh(x_{1},\dots ,x_{M})\rangle $. This expression (\ref%
{bound_02}) is a nontrivial overlap between topologically distinct
geometries (for instance between the geometries depicted in Figure \ref%
{Figure_Bump} (b) and in Figure \ref{Figure_bridge}), and furthermore it is
related to the entanglement entropies as described in (\ref{bound_03}).

The superposition of states corresponding to the same topology gives rise to
a new state that corresponds to a new geometry with a different topology
than these states participating in the superposition. The topology is
changed before and after the superposition. This is similar to the scenario
in \cite{VanRaamsdonk:2010pw}. After the superposition, there is an increase
in the entanglement entropies for the superposed ${\Box _{LM}}$ state and at
the same time a creation of a new bridge structure (see Figure \ref%
{Figure_bridge}).

\begin{figure}[!h]
\centering
\vspace{7pt} \includegraphics[width=0.75\textwidth]{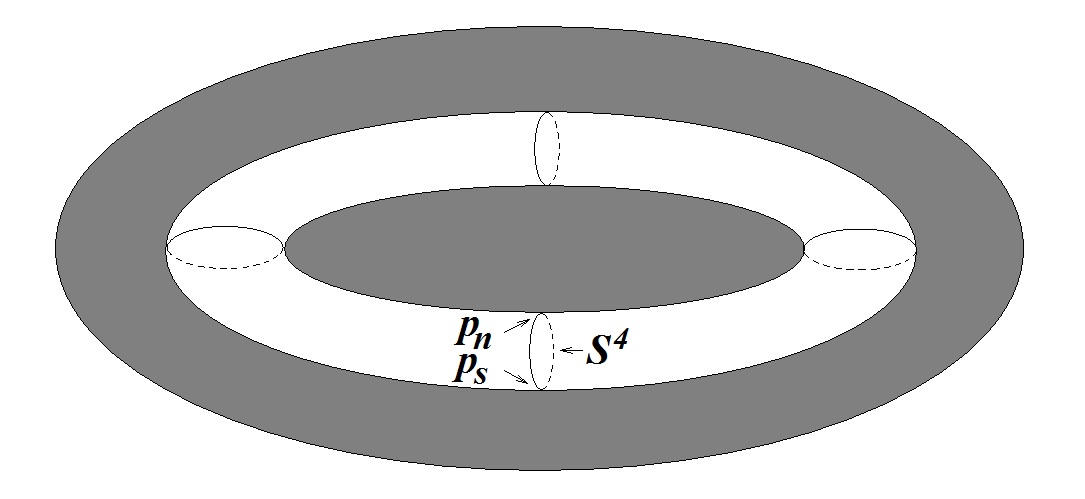} \vspace{1pt%
}
\caption{The bridge connecting the black annulus with the inner black disk.}
\label{Figure_bridge}
\end{figure}

Due to the new bridge, the topology is changed. This bridge connects the
central black droplet with the black annulus (see Figure \ref{Figure_bridge}%
). This bridge is the $S^{3}$ fibration over the white annulus. This is a
bridge connecting two different regions of the same spacetime. It has
topology $S^{1}\times S^{4},$ where $S^{1}$ is the circular direction on the
black and white plane. The $S^{4}$ is formed due to that a $S^{3}$ shrinks
smoothly at the outer and inner edges of the white annulus, and it generates
a nontrivial fourth homology class of the geometry. This bridge connects
with the inner black disk along $S^{1}\times \{p_{\mathrm{n}}\}$ and with
the black annulus\ along $S^{1}\times \{p_{\mathrm{s}}\}$, where $p_{\mathrm{%
n}}$ and $p_{\mathrm{s}}$ are the north pole and south pole of the $S^{4}$.
The geometric property of the bridge depends on $L,M$. The emergence of this
bridge structure is closely related to the entanglement between modes of the
rectangular tableau state. This is reminiscent to the proposal in \cite%
{VanRaamsdonk:2010pw}.

\vspace{0.5cm}
\begin{figure}[h]
\centering
\par
\begin{subfigure}{0.4\textwidth}
	\vspace{2pt} \includegraphics[width=0.75\textwidth]{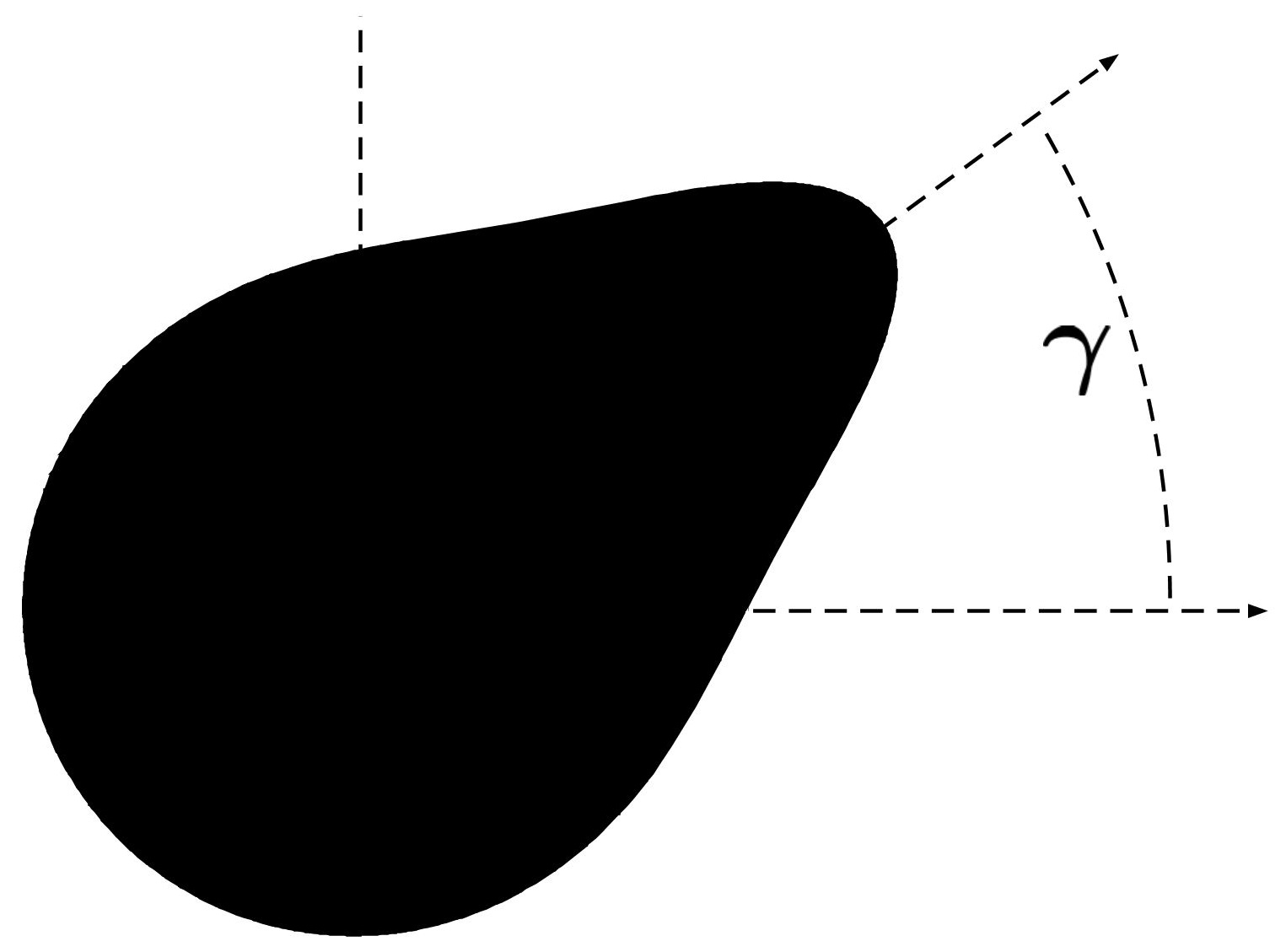}
	\vspace{2pt}
		\caption{Coherent state with a single bump}
		\label{fig:Single_Bump}
	\end{subfigure}
\;\;\;\;\quad
\begin{subfigure}{0.4\textwidth}
	\vspace{3pt} \includegraphics[width=0.6\textwidth]{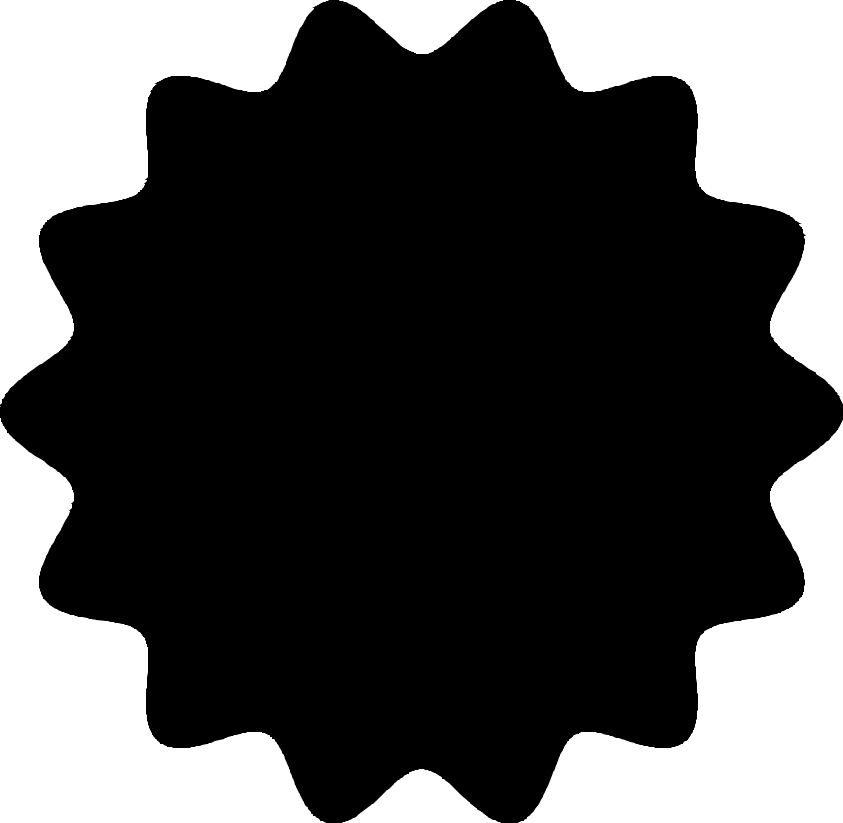}
	\vspace{1pt}
	\caption{Coherent state with multiple bumps that are uniformly distributed}
	\label{fig:multi_Bump}
\end{subfigure}
\caption{Droplet picture of the bubbling geometries for coherent states.
Figure (a) corresponds to a single bump at angle $\protect\gamma$ which
corresponds to the state $B_{+,x}\lvert {0}\rangle$, where $x = r e^{i%
\protect\gamma}$. The parameter $r\in (0,1)$ determines the height of the
bump. The larger the $r$ is, the higher the bump is. And at the limit $r \to
1$, this bump becomes a Dirac delta function. Figure (b) shows $M$ bumps,
where all $x_i$ uniformly distribute around a circle in the complex plane. }
\label{Figure_Bump}
\end{figure}

In deriving this result, we have assumed that all $r_{i}=|x_{i}|$ are equal
for saturation of the bound, which is verified in Appendix B. There is also
a physical interpretation to this. The physical meaning of the above
inequality is to find the upper bound of the overlap of the two states $%
\lvert {\Box _{LM}}\rangle $ and $\prod_{i=1}^{M}B_{+,x_{i}}\lvert {0}%
\rangle $. The geometry corresponding to the state $\lvert {\Box _{LM}}%
\rangle $ is a circular black droplet surrounded by a black annulus. This
geometry has an axial symmetry. On the other hand, we can look at the
geometry corresponding to the state $B_{+,x}\lvert {0}\rangle $. We have a
chiral field $\hat{\phi}(\theta )=\sum_{k>0}(a_{k}\exp (-ik\theta
)+a_{k}^{\dagger }\exp (ik\theta ))$. We can calculate the expectation value
of the chiral field $\phi (\theta )=\langle \hat{\phi}(\theta )\rangle
_{B_{+,x}\lvert {0}\rangle }$, which is
\begin{align}
\phi (\theta )& =\sum_{k=1}^{\infty }x^{k}\exp (-ik\theta )+\bar{x}^{k}\exp
(ik\theta )  \notag \\
& =2\Re (\frac{xe^{-i\theta }}{1-xe^{-i\theta }}).
\end{align}%
The chiral field $\phi (\theta )$ can be regarded as the displacement of the
geometric interface of the two regions, black and white, associated to the
coherent state $B_{+,x}\lvert {0}\rangle $. Writing $x=re^{i\gamma }$, we
can then draw a picture of $\phi (\theta )$, which is like a bump at angle $%
\gamma $ whose height is determined by $r$ (see Figure \ref{Figure_Bump}
(a)). For state $\prod_{i=1}^{M}B_{+,x_{i}}\lvert {0}\rangle $, with $x_{i}$
uniformly distributed around the circle, the corresponding geometry will be
like $M$ bumps with the same height that uniformly distributed along the
angular direction (see Figure \ref{Figure_Bump} (b)). This is the state with
the most possible axial symmetry. Therefore we expect that the maximum will
only be obtained at this state.

\section{Generalized expansion formula}

\label{sec_Generalized superposition formula}

\renewcommand{\theequation}{5.\arabic{equation}} \setcounter{equation}{0} %
\renewcommand{\thethm}{5.\arabic{thm}} \setcounter{thm}{0}

We consider a generalization of the expansion formula for the Young tableau
states in terms of coherent states, in integral representation. Let's first
consider $B_{+}(e^{i\theta })\lvert {0}\rangle =\sum_{n}e^{in\theta }\lvert {%
\triangle _{n}}\rangle $, for the single row states. We write $z=re^{i\theta
}$, then $B_{+}(z)\lvert {0}\rangle =\sum_{n}z^{n}\lvert {\triangle _{n}}%
\rangle $. Then $\lvert {\triangle _{n}}\rangle $ can be written as
\begin{equation}
\lvert {\triangle _{n}}\rangle =\frac{1}{2\pi i}\oint_{\mathcal{C}}\frac{%
\mathrm{d}z}{z}z^{-m}B_{+}(z)\lvert {0}\rangle ,
\end{equation}%
where $\mathcal{C}$ can be any path that encloses $0$.

Hence we can use Fourier transform to represent the single row state $\lvert
{\triangle _{n}}\rangle $ through coherent state
\begin{equation}
\lvert {\triangle _{n}}\rangle =\frac{1}{2\pi }\int_{0}^{2\pi }\mathrm{d}%
\gamma e^{-in\gamma }B_{+,e^{i\gamma }}\lvert {0}\rangle .
\label{single_column_B_02}
\end{equation}%
We then find a generalization of this formula to any Young tableau state.

\begin{thm}
Let $\lambda $ be a Young tableau with $M$ rows (that is, $\lambda =(\lambda
_{1},\lambda _{2},\dots ,\lambda _{M})$). Then the Young tableau state $%
\lvert \lambda \rangle $ can be represented by coherent state $%
B_{+}(x_{1},\dots ,x_{M})\lvert {0}\rangle $ by the following formula
\begin{equation}
\lvert \lambda \rangle =\frac{1}{M!(2\pi )^{M}}\int_{[0,2\pi ]^{M}}\mathrm{d}%
\theta _{1}\dots \mathrm{d}\theta _{M}s_{\lambda }(e^{-i\theta _{1}},\dots
,e^{-i\theta _{M}})\prod_{1\leq i<j\leq M}|e^{i\theta _{i}}-e^{i\theta
_{j}}|^{2}B_{+}(e^{i\theta _{1}},\dots ,e^{i\theta _{M}})\lvert {0}\rangle .
\label{Schur_B_int_02}
\end{equation}
\end{thm}

\begin{proof}
Our proof of the above formula requires using a formula presented in \cite%
{Prob}
\begin{eqnarray}
&&\frac{1}{(2\pi )^{M}}\int_{[0,2\pi ]^{M}}\mathrm{d}\theta _{1}\cdots
\mathrm{d}\theta _{n}J_{\lambda }^{(\alpha )}(e^{i\theta _{1}},\dots
,e^{i\theta _{M}})\overline{J_{\mu }^{(\alpha )}(e^{i\theta _{1}},\dots
,e^{i\theta _{M}})}\prod_{1\leq i<j\leq M}|e^{i\theta _{i}}-e^{i\theta
_{j}}|^{2/{\alpha}}  \notag \\
&=&\delta _{\lambda \mu }\frac{\Gamma (M/\alpha +1)}{\Gamma (1+1/\alpha )^{M}%
}C_{\lambda }(\alpha )\mathcal{N}_{\lambda }^{\alpha }(M)
\end{eqnarray}%
where they used Jack polynomial $J_{\lambda }^{(\alpha )}$ (see Section 4.1
of Ref. \cite{Prob} for details).

For $\alpha =1$, the Jack polynomial $J_{\lambda }^{(\alpha )}$ is a scalar
multiple of the Schur polynomial, $J_{\lambda }^{(1)}=h(\lambda )s_{\lambda
} $ where $h(\lambda )$ is the hook length product. For $\alpha =1$, the
above formula becomes
\begin{equation}
\frac{1}{(2\pi )^{M}}\int_{[0,2\pi ]^{M}}\mathrm{d}\theta _{1}\cdots \mathrm{%
d}\theta _{M}s_{\mu }(e^{i\theta _{1}},\dots ,e^{i\theta _{M}})\overline{%
s_{\lambda }(e^{i\theta _{1}},\dots ,e^{i\theta _{M}})}\prod_{1\leq i<j\leq
M}|e^{i\theta _{i}}-e^{i\theta _{j}}|^{2}=\delta _{\lambda \mu }M! .
\label{Schur_angle_02}
\end{equation}%
Since
\begin{equation}
B_{+}(e^{i\theta _{1}},\dots ,e^{i\theta _{M}})=\sum_{\mu }s_{\mu
}(e^{i\theta _{1}},\dots ,e^{i\theta _{M}})\lvert \mu \rangle ,
\label{B_Schur_02}
\end{equation}%
then we times $\frac{1}{(2\pi )^{M}M!}\overline{s_{\lambda }(e^{i\theta
_{1}},\dots ,e^{i\theta _{M}})}\prod_{1\leq i<j\leq M}|e^{i\theta
_{i}}-e^{i\theta _{j}}|^{2}$ to both side of equation (\ref{B_Schur_02}) and
integrate all $\theta _{l}$, the left hand side is just the right hand side
of (\ref{Schur_B_int_02}), and the right hand side becomes%
\begin{equation}
\sum_{\mu }\delta _{\lambda \mu }\lvert \mu \rangle =\lvert \lambda \rangle .
\end{equation}%
This gives us equation (\ref{Schur_B_int_02}).
\end{proof}

For special case $\lambda =\Box _{LM}$, the above formula becomes
\begin{equation}
\lvert \Box _{LM}\rangle =\frac{1}{M!(2\pi )^{M}}\int_{[0,2\pi ]^{M}}\mathrm{%
d}\theta _{1}\cdots \mathrm{d}\theta _{M}e^{-iL(\theta _{1}+\dots +\theta
_{M})}\prod_{1\leq i<j\leq M}|e^{i\theta _{i}}-e^{i\theta
_{j}}|^{2}B_{+}(e^{i\theta _{1}},\dots ,e^{i\theta _{M}})\lvert {0}\rangle .
\end{equation}%
For special case $\lambda =(n)$, then this formula just gives our previous
formula (\ref{single_column_B_02}).

This is a generalized expansion formula in integral representation. The
relation between this formula and $B_{+}(x_{1},x_{2},\dots )=\sum_{\lambda
}s_{\lambda }\lvert \lambda \rangle $ is just like the relation between
formula $\lvert {\triangle _{n}}\rangle =\frac{1}{2\pi }\int_{0}^{2\pi }%
\mathrm{d}\gamma e^{-in\gamma }B_{+,e^{i\gamma }}\lvert {0}\rangle $ and $%
B_{+}(x)\lvert {0}\rangle =\sum_{n}x^{n}\lvert {\triangle _{n}}\rangle $,
where we take inverse Fourier transform of $B_{+}(x)\lvert {0}\rangle $ to
get $\lvert {\triangle _{n}}\rangle $. This is some kind of generalized
Fourier transform in the case of symmetric polynomial.

There is also an analogous formula for $B_{-}(-x_{1},-x_{2},\dots )$. The
dual of a coherent state $B_{+}(x_{1},x_{2},\dots )$ is just $%
B_{-}(-x_{1},-x_{2},\dots )$. Then we just need to take the dual of the
formula to give%
\begin{equation}
\lvert \lambda ^{T}\rangle =\frac{1}{M!(2\pi )^{M}}\int_{[0,2\pi ]^{M}}%
\mathrm{d}\theta _{1}...\mathrm{d}\theta _{M}s_{\lambda }(e^{-i\theta
_{1}},...,e^{-i\theta _{M}})\prod_{1\leq i<j\leq M}|e^{i\theta
_{i}}-e^{i\theta _{j}}|^{2}B_{-}(-e^{i\theta _{1}},...,-e^{i\theta
_{M}})\lvert {0}\rangle,  \label{Schur_B_-_int_02}
\end{equation}%
where $M$ is the number of columns of the Young tableau $\lambda ^{T}$. This
is easily proved in the same way as in the above proof. The equation (\ref%
{Schur_B_int_02}) is dual to (\ref{Schur_B_-_int_02}).

It is very interesting that the paper \cite{Prob} is on Probability Theory
and actually discuss Random matrix. This is not a surprise because our
system also has a matrix theory description. And it's also very interesting
that some formulas given in \cite{Prob} will be useful in our system.

\section{Squeezed states and multi-mode entangled states from Young tableau
states}

\label{sec_Squeezed states, multi-mode entangled states from Young tableau
states}\renewcommand{\theequation}{6.\arabic{equation}} %
\setcounter{equation}{0}

Here we discuss other states that are of interest, including the two-mode
squeezed states and multi-mode entangled states. In quantum information
theory and quantum optics, these states naturally arise for their enormous
applications. Putting these states into our setting, we find other
interesting results. We compute their overlaps with the Young tableau
states, and further make expansions of them in terms of Young tableau states.

Let us first consider the two-mode squeezed states
\begin{equation}
\lvert \mathrm{Squ}_{kk^{\prime }}\rangle =\exp \left[ \mu (a_{k}^{\dagger
}a_{k^{\prime }}^{\dagger }-a_{k}a_{k^{\prime }})\right] \lvert {0}\rangle .
\end{equation}%
One mode is created by $a_{k}^{\dagger }$ and belongs to the Hilbert space $%
\mathcal{H}_{k}$. At the same time, the other mode is created by $%
a_{k^{\prime }}^{\dagger }$ and belongs to the Hilbert space $\mathcal{H}%
_{k^{\prime }}$. For simplicity, we consider the parameter $\mu $ to be
real. After using the commutation relations
\begin{equation}
\left[ \frac{a_{k_{1}}}{\sqrt{k_{1}}},~\frac{a_{k_{2}}^{\dagger }}{\sqrt{%
k_{2}}}\right] =\delta _{k_{1}k_{2}},
\end{equation}%
we have that%
\begin{equation}
\lvert \mathrm{Squ}_{kk^{\prime }}\rangle =\left( 1-\tanh ^{2}(\sqrt{%
kk^{\prime }}\mu )\right) ^{\frac{1}{2}}\sum\limits_{l=0}^{\infty }\left(
\tanh (\sqrt{kk^{\prime }}\mu )\right) ^{l}\frac{1}{l!(\sqrt{kk^{\prime }}%
)^{l}}\lvert t_{k}^{l}\rangle \otimes \lvert t_{k^{\prime }}^{l}\rangle .
\end{equation}%
We denote $q=\tanh (\sqrt{kk^{\prime }}\mu ).$

We use the generating function $Z_{j}(t)=\langle \exp (iH_{j}t)\rangle
_{\rvert \text{Squ}_{kk^{\prime }}\rangle }=\langle \exp (ia_{j}^{\dagger
}a_{j}t)\rangle _{\rvert \text{Squ}_{kk^{\prime }}\rangle }$, where the
expectation value is taken on squeezed state $\rvert \text{Squ}_{kk^{\prime
}}\rangle $. For $j\neq k,k^{\prime }$, it is simple because $\rvert \text{%
Squ}_{kk^{\prime }}\rangle $ do not contain a $j$ mode. For $j=k$, consider
\begin{equation}
\exp (iH_{k}t)\rvert \text{Squ}_{kk^{\prime }}\rangle =(1-q^{2})^{\frac{1}{2}%
}\sum_{l=0}^{\infty }\frac{q^{l}}{l!(\sqrt{kk^{\prime }})^{l}}\exp
(ilkt)\lvert t_{k}^{l}\rangle \otimes \lvert t_{k^{\prime }}^{l}\rangle .
\end{equation}%
Thus we have
\begin{equation}
Z_{k}(t)=\langle \exp (iH_{k}t)\rangle _{\rvert \text{Squ}_{kk^{\prime
}}\rangle }=(1-q^{2})\sum_{l=0}^{\infty }q^{2l}\exp (ilkt)=\frac{1-q^{2}}{%
1-q^{2}\exp (ikt)}.
\end{equation}%
Likewise, for $j=k^{\prime }$, we get
\begin{equation}
Z_{k^{\prime }}(t)=\frac{1-q^{2}}{1-q^{2}\exp (ik^{\prime }t)}.
\end{equation}%
This is just a statement that for mode $k$ or $k^{\prime }$, the probability
distribution is $p_{l}\propto q^{2l}$. This can also be shown by directly
calculating the reduced density matrix through tracing.

The overlaps between the above squeezed states and the Young tableau states
are hence%
\begin{equation}
\langle \lambda \lvert \text{Squ}_{kk^{\prime }}\rangle =\left(
1-q^{2}\right) ^{\frac{1}{2}}\frac{q^{l}}{l!(\sqrt{kk^{\prime }})^{l}}\chi
_{\lambda }(\vec{w})\rvert _{\substack{ w_{k}=w_{k^{\prime }}=l;  \\ %
w_{p}=0,p\neq k,k^{\prime }}}.
\end{equation}%
And $\lambda \vdash l(k+k^{\prime })$. \

It's easy to see that the above squeezed states have the inner products
\begin{equation}
\langle \mathrm{Squ}_{k_{1}k_{2}}\rvert \mathrm{Squ}_{k_{3}k_{4}}\rangle
=\delta _{k_{1}k_{3}}\delta _{k_{2}k_{4}}+\delta _{k_{1}k_{4}}\delta
_{k_{2}k_{3}}.
\end{equation}%
As the Young tableau states are orthonormal \cite{Corley:2001zk}, the
squeezed states are the linear combinations of the Young tableau states as%
\begin{equation}
\rvert \mathrm{Squ}_{kk^{\prime }}\rangle =\left( 1-q^{2}\right) ^{\frac{1}{2%
}}\sum_{l=0}^{\infty }\sum_{\lambda \vdash l(k+k^{\prime })}\left( \frac{%
q^{l}}{l!(\sqrt{kk^{\prime }})^{l}}\chi _{\lambda }(\vec{w})\rvert
_{\substack{ w_{k}=w_{k^{\prime }}=l;  \\ w_{p}=0,p\neq k,k^{\prime }}}%
\right) \lvert \lambda \rangle .
\end{equation}

Taking a limit $\sqrt{kk^{\prime }}\mu \rightarrow \infty $, or equivalently
$q\rightarrow 1$, the above two-mode squeezed states become maximally
entangled states, or EPR states $\lvert \mathrm{EPR}\rangle =\lim_{\sqrt{%
kk^{\prime }}\mu \rightarrow \infty }\lvert \mathrm{Squ}_{kk^{\prime
}}\rangle $, which are%
\begin{equation}
\lvert \mathrm{EPR}\rangle =\mathcal{N}^{-\frac{1}{2}}\sum\limits_{l=0}^{%
\infty }\frac{1}{l!(\sqrt{kk^{\prime }})^{l}}\lvert t_{k}^{l}\rangle \otimes
\lvert t_{k^{\prime }}^{l}\rangle ,  \label{EPR_02}
\end{equation}%
where $\mathcal{N}^{-\frac{1}{2}}$ is a normalization factor. The
normalization factor can be understood as follows. One can take an
infinitesimal positive cutoff $\epsilon \rightarrow 0$, such that $%
1-q=\epsilon $ and $\mathcal{N}=\frac{1}{2\epsilon }$. The squeezed state $%
\lvert \mathrm{Squ}_{kk^{\prime }}\rangle $ with $q$ close to 1, is a good
approximation to the EPR state. These are the entangled states in $\mathcal{H%
}_{k}\otimes \mathcal{H}_{k^{\prime }}$, in which the states $\lvert
l\rangle _{k}=\frac{1}{\sqrt{l!k^{l}}}\lvert t_{k}^{l}\rangle $ in mode $k$,
are entangled with the states $\lvert l\rangle _{k^{\prime }}=\frac{1}{\sqrt{%
l!(k^{\prime })^{l}}}\lvert t_{k^{\prime }}^{l}\rangle $ in mode $k^{\prime
} $.$~$Consider mode $k$ as a IR mode with $k$ very small, and consider mode
$k^{\prime }$ as a UV mode with $k^{\prime }$ very big. Then this state is
an entangled state between a IR mode and a UV mode. Integrating out the UV
mode, gives rise to a reduced density matrix of the IR mode.

The EPR states and the Young tableau states have overlaps
\begin{equation}
\langle \lambda \rvert \mathrm{EPR}\rangle =\mathcal{N}^{-\frac{1}{2}}\frac{1%
}{l!(\sqrt{kk^{\prime }})^{l}}\chi _{\lambda }(\vec{w})\rvert _{\substack{ %
w_{k}=w_{k^{\prime }}=l;  \\ w_{p}=0,p\neq k,k^{\prime }}} .
\end{equation}%
Thus the EPR states can be written as the linear combinations of the Young
tableau states as
\begin{equation}
\rvert \mathrm{EPR}\rangle =\mathcal{N}^{-\frac{1}{2}}\sum_{l=0}^{\infty
}\sum_{\lambda \vdash l(k+k^{\prime })}\left( \frac{1}{l!(\sqrt{kk^{\prime }}%
)^{l}}\chi _{\lambda }(\vec{w})\rvert _{\substack{ w_{k}=w_{k^{\prime }}=l;
\\ w_{p}=0,p\neq k,k^{\prime } }}\right) \lvert \lambda \rangle .
\end{equation}

Consider the state
\begin{equation}
\sqrt{1-q^{2}}(\sum_{l=0}^{\infty }q^{l}\lvert l\rangle _{k}\lvert l\rangle
_{k^{\prime }}\lvert l\rangle _{k^{\prime \prime }})=\sqrt{1-q^{2}}%
(\sum_{l=0}^{\infty }q^{l}\frac{1}{\sqrt{(l!)^{3}(kk^{\prime }k^{\prime
\prime })^{l}}}\lvert {t_{k}^{l}}\rangle \lvert {t_{k^{\prime }}^{l}}\rangle
\lvert {t_{k^{\prime \prime }}^{l}}\rangle ),
\end{equation}%
while the GHZ state corresponds to the $q\rightarrow 1$ limit. In our
system, a three mode entangled GHZ state, should generalize the above state
as
\begin{equation}
\lvert \mathrm{GHZ}\rangle =\tilde{\mathcal{N}}^{-\frac{1}{2}%
}(\sum_{l=0}^{\infty }\lvert l\rangle _{k}\lvert l\rangle _{k^{\prime
}}\lvert l\rangle _{k^{\prime \prime }}).  \label{GHZ_01}
\end{equation}%
This is an entangled state in $\mathcal{H}_{k}\otimes \mathcal{H}_{k^{\prime
}}\otimes \mathcal{H}_{k^{\prime \prime }}$. Then we trace out any two
modes, and this gives a density matrix
\begin{equation}
\hat{\rho}=(1-q^{2})\sum_{l=0}^{\infty }q^{2l}\lvert l\rangle _{k}\langle {l}%
\rvert _{k}~.
\end{equation}%
The entropy for this density matrix is $s=-(\log (1-q^{2})+\frac{2q^{2}\log
(q)}{1-q^{2}})$, and taking $q\rightarrow 1$ limit gives $s\rightarrow
\infty $.

The GHZ states and the Young tableau states have overlaps
\begin{equation}
\langle \lambda \rvert \mathrm{GHZ}\rangle =\tilde{\mathcal{N}}^{-\frac{1}{2}%
}\frac{1}{\sqrt{(l!)^{3}(kk^{\prime }k^{\prime \prime })^{l}}}\chi _{\lambda
}(\vec{w})\rvert _{\substack{ w_{k}=w_{k^{\prime }}=w_{k^{\prime \prime
}}=l;  \\ w_{p}=0,p\neq k,k^{\prime },k^{\prime \prime }}}.  \label{GHZ_02}
\end{equation}%
The GHZ states can thus be written as the linear combinations of the Young
tableau states as\vspace{1pt}%
\begin{equation}
\rvert \mathrm{GHZ}\rangle =\tilde{\mathcal{N}}^{-\frac{1}{2}%
}\sum_{l=0}^{\infty }\sum_{\lambda \vdash l(k+k^{\prime }+k^{\prime \prime
})}\left( \frac{1}{\sqrt{(l!)^{3}(kk^{\prime }k^{\prime \prime })^{l}}}\chi
_{\lambda }(\vec{w})\rvert _{\substack{ w_{k}=w_{k^{\prime }}=w_{k^{\prime
\prime }}=l;  \\ w_{p}=0,p\neq k,k^{\prime },k^{\prime \prime }}}\right)
\lvert \lambda \rangle  \label{GHZ_03}
\end{equation}%
and $\lambda \vdash l(k+k^{\prime }+k^{\prime \prime })$.

It is easy to generalize the state in (\ref{GHZ_01}), (\ref{GHZ_03}) to a $m$%
-mode GHZ state with $m>3$,
\begin{eqnarray}
\lvert \mathrm{GHZ}\rangle _{m} &=&\tilde{\mathcal{N}}_{m}^{-\frac{1}{2}%
}(\sum_{l=0}^{\infty }\lvert l\rangle _{k_{1}}\lvert l\rangle _{k_{2}}\dots
\lvert l\rangle _{k_{m}})  \notag \\
&=&\tilde{\mathcal{N}}_{m}^{-\frac{1}{2}}\sum_{l=0}^{\infty }\sum_{\lambda
\vdash l(k_{1}+k_{2}+\dots +k_{m})}\left( \frac{1}{\sqrt{(l!)^{m}(k_{1}k_{2}%
\dots k_{m})^{l}}}\chi _{\lambda }(\vec{w})\rvert _{\substack{ %
w_{k_{1}}=w_{k_{2}}=\dots =w_{k_{m}}=l;  \\ w_{p}=0,p\neq k_{1},k_{2},\dots
,k_{m}}}\right) \lvert \lambda \rangle .  \notag \\
&&
\end{eqnarray}

We can also calculate the inner product $\langle \text{Squ}_{kk^{\prime
}}\lvert B_{+}(x_{1},\dots )\rvert 0\rangle $, which is just the symmetric
polynomial associated with $\rvert \text{Squ}_{kk^{\prime }}\rangle $. We
get the inner product%
\begin{equation}
\langle \text{Squ}_{kk^{\prime }}\lvert B_{+}(x_{1},x_{2},\dots )\rvert
0\rangle =(1-q^{2})^{\frac{1}{2}}\sum_{l=0}^{\infty }\frac{q^{l}}{l!(\sqrt{%
kk^{\prime }})^{l}}p_{k}^{l}p_{k^{\prime }}^{l}=(1-q^{2})^{\frac{1}{2}}\exp (%
\frac{q}{\sqrt{kk^{\prime }}}p_{k}p_{k^{\prime }}),
\end{equation}%
where $p_{k}=x_{1}^{k}+x_{2}^{k}+\dots $ , and $p_{k^{\prime
}}=x_{1}^{k^{\prime }}+x_{2}^{k^{\prime }}+\dots $. Hence,
\begin{equation}
\langle \text{Squ}_{kk^{\prime }}\lvert B_{+}(x_{1},x_{2},\dots )\rvert
0\rangle =(1-q^{2})^{\frac{1}{2}}\exp (\frac{q}{\sqrt{kk^{\prime }}}%
(x_{1}^{k}+x_{2}^{k}+\dots )(x_{1}^{k^{\prime }}+x_{2}^{k^{\prime }}+\dots
)).
\end{equation}%
This result is just a special case of the fact that the overlap of any state
with $B_{+}(x_{1},\dots )\rvert 0\rangle $ is the symmetric polynomial
associated with that state. See proposition \ref{proposition_coh_02}.

Consider the field
\begin{equation}
\hat{\phi}(\theta )=\sum_{m>0}a_{m}\exp (-im\theta )+a_{m}^{\dagger }\exp
(im\theta ).
\end{equation}%
First the expectation value of this field on the above squeezed states is
zero, $\langle \hat{\phi}(\theta )\rangle =0$. Then we compute the variance $%
\langle $Squ$_{kk^{\prime }}\lvert :\hat{\phi}^{2}(\theta ):\rvert $Squ$%
_{kk^{\prime }}\rangle $, where $:\hat{\phi}^{2}(\theta ):$ is the normal
ordering of $\hat{\phi}^{2}(\theta )$. The variance is
\begin{equation}
\langle \text{Squ}_{kk^{\prime }}\lvert :\hat{\phi}^{2}(\theta ):\rvert
\text{Squ}_{kk^{\prime }}\rangle =\frac{2q}{(1-q^{2})}(q(k+k^{\prime })+2%
\sqrt{kk^{\prime }}\cos ((k+k^{\prime })\theta )).
\end{equation}%
The detailed derivation of the above expression is in Appendix \ref%
{appendix_variance}. This result also reveals that the squeezed state is not
rotationally symmetric. We also calculate other quantities like $\langle {%
\hat{N}}_{j}\rangle =\langle a_{j}^{\dagger }a_{j}\rangle $ and $\langle
a_{m}^{\dagger }a_{j}\rangle ,$
\begin{equation}
\langle a_{m}^{\dagger }a_{j}\rangle =%
\begin{cases}
\frac{q^{2}}{(1-q^{2})}k & ~~~m=j=k \\
\frac{q^{2}}{(1-q^{2})}k^{\prime } & ~~~m=j=k^{\prime } \\
0 & ~~~\text{others}%
\end{cases}%
.
\end{equation}

The squeezed state is interesting that it tells us that we can create a EPR
pair by squeezing the vacuum. They and the multi-mode entangled states can
be expanded by Young tableau states in the setup here and are very
interesting states in gauge/gravity correspondence. For discussions of
squeezed states in quantum information theory and quantum optics, see for
example \cite{Quantum information,Squeezed01,Squeezed02,Squeezed03}. GHZ
states are also very important in quantum information theory \cite%
{Horodecki:2009zz}.

\section{Discussion}

\renewcommand{\theequation}{7.\arabic{equation}} \setcounter{equation}{0}

\label{sec_Discussion}

We computed a momentum space version of the entanglement spectrum and
entanglement entropy of Young tableau states and one-point functions on
Young tableau states. The Young tableau states are not direct product
states, and they have non-zero entanglement between modes. The entanglement
spectrum and entanglement entropy of general Young tableau states are
obtained. We have also computed the generating functions for one-point
functions on Young tableau states. These physical quantities in the field
theory side are used to measure the topology of the dual spacetime
geometries, such as the number of annuli in the geometries and the existence
of bridge structures which connect different regions of the same spacetime.
Our results indicate that the emergence of the bridge structure is closely
related to the entanglement between modes of the Young tableau states.

On one hand, we can expand a coherent state above the vacuum as the linear
combination of Young tableau states through our explicit expression for the
inner products between coherent states and general Young tableau states.
This is an analog of Fourier transform. On the other hand, the Young tableau
states can be obtained by superposition of coherent states, and we further
presented an integral formula as an inverse transform. Thus we get two sets
of formulas, one is to express coherent states by Young tableau states and
the second is to express Young tableau states by coherent states, and the
relation between these two sets of formulas is an analogy to the Fourier
transform and its inverse transform. These formulas are also of mathematical
interest. Form a physical point of view, we are particularly interested in
expressing the rectangular Young tableau state by coherent states, since the
superposed geometry (see Figure \ref{Figure_droplet} (a)) has a different
topology than the original geometries (see Figure \ref{Figure_droplet} (b,
c)) participating in the superposition. At the same time, the bridge
structure emerged after the topology change, and this is related to the
entanglement between modes of the Young tableau states. We can then
generalize the case of rectangular Young tableau states to other Young
tableau states. This further implies that we can superpose topology trivial
states to get states with complicated geometry and topology in the gravity
side.

One important feature of our system is that these states with different
topologies live in the same Hilbert space, hence one can concretely study
the transition amplitudes between different states from the dual quantum
mechanical system. We analyzed the overlaps between Young tableau states and
coherent states, and carried out in detail for the rectangular tableaux,
corresponding to the geometries with one black annulus. As shown in Section %
\ref{sec_Bound of overlap and entanglement entropy}, we have a refined bound
for the overlap between coherent states and a rectangular Young tableau
state. The overlap between two states differed by this topology change is
exponentially suppressed. Hence to produce a topologically distinct geometry
by superposing coherent states dual to geometries with a trivial topology,
it requires at least an exponentially large number of states in the
superposition. This is essentially a non-perturbative effect in quantum
gravity. It is further found that the norm squared of the overlaps is
bounded above by the inverse powers of the exponential of the entanglement
entropies. Our results put into firmer footing the insights and observations
in \cite{Berenstein:2017abm,Berenstein:2016pcx,Berenstein:2016mxt}.
Incidentally, we also find that the overlap of any state with the coherent
state defined in Section 3, is a symmetric function associated with that
state. This greatly generalized the results in \cite{Berenstein:2017abm}.
And it provides us an approach to analyze the overlap between a coherent
state and an arbitrary state with more complicated topology, like those with
more white rings, which correspond to Young tableaux with more long edges.
Since one can understand the norm squared of the overlap as the transition
probability between a topology trivial state and a state with a distinct
topology, it would be valuable to explore the relation between the
transition probabilities and the geometric and topological properties of the
states. Also, our analysis related the characters of the symmetric groups
\cite{Sagan,James Kerber} to the topologies of the bubbling geometries. We
hope that more physical intuition will give further mathematical insights
into related subjects.

Here, these exponentially large number of states participating in the
superposition have caused the topology change. On the other hand, the
situations with a small number of states participating in the superposition
can be different from the situations with an exponentially large number of
states in the superposition. The former cases would imply relatively bigger
overlaps between individual states. Their differences are also pointed out
in \cite{Almheiri:2016blp} and \cite{Nomura:2016aww} in closely related
discussions.

The theorem in Section 2 and a proposition in Section 3 are analogous to
giant/dual-giant duality. These are dualities between giant gravitons
wrapping AdS directions and dual giant gravitons wrapping internal
directions. They have also appeared in other context, see for example \cite%
{Lin:2012ey,Bissi:2011dc,Caputa:2012yj,Carlson:2011hy,deMelloKoch:2012ck,Koch:2011hb}
in two-point and three-point functions of giant gravitons. Moreover, bearing
the different aspects to look at our system in mind, this giant/dual-giant
duality can have interpretations in other ways. In the droplet picture, this
duality is a particle/hole duality. Since each state is associated with a
symmetric function, this duality can be explained in terms of an involution
on the ring of symmetric functions.

Topology change in bubbling geometries were also discussed in \cite%
{Horava:2005pv,Mosaffa:2006qk}. Other aspects of describing topology on the
gravity side from the field theory side have been put forward, in \cite%
{Chen:2007gh,Koch:2008ah,Lin:2010sba} by using correlation functions for the
states dual to strings on bubbling geometries, in \cite{Brown:2006zk} by
using correlation functions for the states undergoing topology change, and
in \cite{Diaz:2015tda} by using probability theory on the graphs of
representations \cite{Borodin Olshanski}.

To sum over different topologies and geometries are important issues in
quantum gravitational theories, see for example \cite{Berenstein:2017abm,
Horowitz:2006ct,Jafferis:2017tiu,Kiefer:2005uk}. Various other similar
geometries in the context of string theory and quantum gravity have been
analyzed, see for example [67$-$76, 17] and their related discussions. Our
approach may also be related to fuzzball proposal \cite{Mathur:2005ai} and
to 2d Yang-Mills \cite{Dijkgraaf:2005bp}. It would also be good to
understand in more detail the relation to the scenarios of building
spacetime geometries, as proposed in for example \cite%
{VanRaamsdonk:2010pw,Maldacena:2013xja,Rangamani:2016dms}.

Two mode squeezed states and multi-mode entangled states are also discussed,
and the two mode maximally entangled state, the EPR state, is a particular
limit of the squeezed state. They have similarities with those states
appeared in quantum optics and quantum information theory \cite{Quantum
information,Squeezed01,Squeezed02,Squeezed03}. However, our setup provides
another framework to explore their properties. By computing their overlaps
with Young tableau states, we expanded them as linear combinations of Young
tableau states. Maybe we can consider the EPR state in Section 6 as similar
to creating a squeezed state string connecting between mode $k$ and $%
k^{\prime }$. We think that it would be similar to the situations in the
ER=EPR proposal \cite{Maldacena:2013xja}. This proposal \cite%
{Maldacena:2013xja} has related the geometry side \cite{Einstein:1935tc} to
the side of quantum mechanics system \cite{Einstein:1935rr} and conjectured
that entangled black holes are connected by worm hole. It may be interesting
to understand more the relations to this proposal for the squeezed states
and EPR states discussed in Section 6.

The approach here describes measuring the topologies of spacetime from the
dual field theory side, and therefore will be interesting and useful for
understanding the emergence of spacetime structures, see for example \cite%
{Rangamani:2016dms,VanRaamsdonk:2010pw,Horowitz:2006ct,Koch:2009gq}. These
geometries are very explicit and they serve as a good laboratory to perform
quantitative calculations and predictions. Moreover, the system is UV
finite, since it has UV completion in string theory. Various other
discussions on entanglement entropies with bubbling geometries are recently
in for example \cite{Balasubramanian:2017hgy} and it would be interesting to
see the relation to the discussions here.

The entanglement entropies here arise after partial tracing out other
momentum modes living in the momentum space version of the Hilbert space
decomposition. The subsystem in this case, is a region in the momentum
space. Consider high energy UV modes entangled with low energy IR modes.
Then tracing out the high energy modes gives a reduced density matrix for
the low energy modes. Hence in these cases, physics at low energy can still
be sensitive to the details of the physics at high energy. This momentum
space version of the entanglement entropy is similar but slightly different
from the usual real space version of the entanglement entropy, where the
real space version of the Hilbert space decomposition is used and the
subsystems are domains in real space. The two are related by a different
decomposition of the Hilbert space.

\section*{Acknowledgments}

We would like to thank B. Czech, R. de Mello Koch, Q. T. Li, J. Maldacena,
R. Miao, S. Ramgoolam, J. Shock, J. Simon, H. Verlinde, S.-T. Yau, and J. Wu
for discussions and communications. The work was supported in part by NSF
grant DMS-1159412, NSF grant PHY-0937443 and NSF grant PHY-1306313, and in
part by YMSC and Tsinghua University.



\appendix

\section{Overlap of coherent states and Young tableau states}

\renewcommand{\theequation}{A.\arabic{equation}} \setcounter{equation}{0}

\label{appendix_inner product}

The norm-squared $\Vert B_{+,x_{1}}B_{+,x_{2}}\lvert {0}\rangle \Vert ^{2}$
was computed in Eq. (4.44) of \cite{Berenstein:2017abm}. We can generalize
this formula to the inner products of arbitrarily many $B_{+,x_{i}}$, and we
have the following result
\begin{equation}
\langle {0}\rvert (\prod_{i}B_{+,y_{i}}^{\dagger
})(\prod_{i}B_{+,x_{i}})\lvert {0}\rangle =\prod_{i,j}\frac{1}{1-x_{i}\bar{y}%
_{j}}.  \label{inner_product_01}
\end{equation}%
\begin{proof}
	Write
\begin{equation}\prod_{i}B_{+,x_i} \lvert{0}\rangle = \sum_{\lambda} s_{\lambda}(x_1,x_2,\dots) \lvert{\lambda}\rangle,
\end{equation}
then insert it in the above formula
	\begin{align}
	\langle{0}\rvert(\prod_iB_{+,y_i}^\dagger)(\prod_{i}B_{+,x_i})\lvert{0}\rangle & = \langle{0}\rvert(\prod_iB_{+,y_i}^\dagger)\sum_{\lambda} s_{\lambda}(x_1,x_2,\dots) \lvert{\lambda}\rangle   \nonumber\\
	& =  \sum_{\lambda} s_{\lambda}(x_1,x_2,\dots) \langle{0}\rvert(\prod_iB_{+,y_i}^\dagger)\lvert {\lambda} \rangle \nonumber\\
	& = \sum_{\lambda} s_{\lambda}(x_1,x_2,\dots) s_{\lambda}(\bar{y}_1,\bar{y}_2,\dots) \nonumber\\
	& = \prod_{i,j}\frac{1}{1 - x_i\bar{y}_j}.
	\end{align}
In the last line, we used the Cauchy identity.
\end{proof}

Consider the case that we only have one variable $x_{1}$ and $y_{1}=x_{1}$,
then the above formula gives Eq. (4.6) in \cite{Berenstein:2017abm}, and for
two variables $x_{1}=y_{1},\;x_{2}=y_{2}$, our formula gives Eq. (4.44) in
\cite{Berenstein:2017abm}.

It is easy to see the norm of the $\prod_{i}B_{+,x_{i}}\lvert {0}\rangle $
from (\ref{inner_product_01}), when identifying $y_{i}$ with $x_{i}$, hence
\begin{equation}
\parallel \prod_{i}B_{+,x_{i}}\lvert {0}\rangle \parallel ^{2}=\prod_{i,j}%
\frac{1}{1-x_{i}\bar{x}_{j}}.
\end{equation}

There is a dual version of the above formula for operators $B_{-,x_{i}}$,
which gives the same result as above
\begin{equation}
\langle {0}\rvert (\prod_{i}B_{-,y_{i}}^{\dagger
})(\prod_{i}B_{-,x_{i}})\lvert {0}\rangle =\prod_{i,j}\frac{1}{1-x_{i}\bar{y}%
_{j}}.
\end{equation}%
\begin{proof}
	According to Proposition \ref{prop_coh_03},
\begin{equation}
\langle{\lambda^T}\rvert\prod_{i}B_{-,-x_i}\lvert{0} \rangle = s_{\lambda},
\end{equation}
therefore
\begin{equation}
\prod_{i}B_{-,x_i}\lvert{0}\rangle = \sum_{\lambda} s_{\lambda}(-x_1,-x_2,\dots)\lvert{\lambda^T} \rangle.
\end{equation}
Inserting this into the inner product,
	\begin{align}
	\langle{0}\rvert(\prod_iB_{-,y_i}^\dagger)(\prod_{i}B_{-,x_i})\lvert{0}\rangle &= \langle{0}\rvert(\prod_iB_{-,y_i}^\dagger)\sum_{\lambda} s_{\lambda}(-x_1,-x_2,\dots)\lvert{\lambda^T}\rangle \nonumber\\
	&=\sum_{\lambda} s_{\lambda}(-x_1,-x_2,\dots)\langle{0}\rvert(\prod_iB_{-,y_i}^\dagger) \lvert{\lambda^T} \rangle \nonumber\\
	&=\sum_{\lambda} s_{\lambda}(-x_1,-x_2,\dots)s_{\lambda}(-\bar{y}_1,-\bar{y_2},\dots) \nonumber\\
	&= \prod_{i,j}\frac{1}{1 - x_i\bar{y}_j}.
	\end{align}
\end{proof}

There is also a formula for inner product of $B_{+}$ with $B_{-}$, which is
as follows
\begin{equation}
\langle {0}\rvert (\prod_{i}B_{-,y_{i}}^{\dagger
})(\prod_{i}B_{+,x_{i}})\lvert {0}\rangle =\prod_{i,j}(1-x_{i}\bar{y}_{j}).
\end{equation}%
\begin{proof}
\begin{align}
\langle{0}\rvert(\prod_iB_{-,y_i}^\dagger)(\prod_{i}B_{+,x_i})\lvert {0} \rangle & = \langle{0}\rvert(\prod_iB_{-,y_i}^\dagger)\sum_{\lambda} s_{\lambda}(x_1,x_2,\dots) \lvert {\lambda} \rangle  \nonumber\\
& = \sum_{\lambda} s_{\lambda}(x_1,x_2,\dots) \langle{0}\rvert(\prod_iB_{-,y_i}^\dagger)\lvert{\lambda}\rangle  \nonumber\\
& = \sum_{\lambda} s_{\lambda}(x_1,x_2,\dots) s_{\lambda^T}(-\bar{y}_1,-\bar{y}_2,\dots)  \nonumber\\
& = \prod_{i,j}(1 - x_i\bar{y}_j).
\end{align}
\end{proof}

These formulas also provide the normalizations of the two-point functions
for coherent states.

\section{Bound of overlap}

\renewcommand{\theequation}{B.\arabic{equation}} \setcounter{equation}{0} %
\renewcommand{\theprop}{B.\arabic{prop}} \setcounter{prop}{0}

\label{appendix_bound of overlap}

As described in Section \ref{sec_Bound of overlap and entanglement entropy},
we will analyze the function%
\begin{equation}
f(x_{1},x_{2},\dots ,x_{M}):=\frac{|\langle {\lambda }\rvert
\prod_{i=1}^{M}B_{+,x_{i}}\lvert {0}\rangle |^{2}}{\Vert
\prod_{i=1}^{M}B_{+,x_{i}}\lvert {0}\rangle \Vert ^{2}}.
\end{equation}%
We consider the case with the ${\Box _{LM}}$ state which corresponds to $n_{%
\mathrm{anti-edge}}=1$.

\begin{prop}
For a rectangular Young tableau with $M$ rows and $L$ columns, denoted by ${%
\Box _{LM}}$, and for coherent states $\prod_{i=1}^{M}B_{+,x_{i}}\lvert {0}%
\rangle $ with arbitrary $\{x_{i}\}$, the supremum of their normalized inner
products is given by%
\begin{equation}
\sup_{\{x_{i}\}}\left\vert \frac{\langle {\Box _{LM}}\rvert
\prod_{i=1}^{M}B_{+,x_{i}}\lvert {0}\rangle }{\Vert
\prod_{i=1}^{M}B_{+,x_{i}}\lvert {0}\rangle \Vert }\right\vert ^{2}=(\frac{M%
}{L+M})^{M}(\frac{L}{L+M})^{L}.  \label{prop_inner_product_01}
\end{equation}
\end{prop}

\newpage
\begin{proof}
Writing $x_{i}=r_{i}e^{\theta _{i}}$, and inserting them into the above
formula we have $f=f_{1}f_{2}=\prod_{i,j}(1-r_{i}r_{j}e^{i(\theta
_{i}-\theta _{j})})\times (\prod_{i=1}^{M}r_{i}^{2})^{L}$, we can then write
the first term as
\begin{align}
f_{1}& :=\prod_{i,j}(1-r_{i}r_{j}e^{i(\theta _{i}-\theta _{j})})  \notag \\
& =\prod_{i}(1-r_{i}^{2})\prod_{1\leq i<j\leq M}(1-2r_{i}r_{j}\cos (\theta
_{i}-\theta _{j})+r_{i}^{2}r_{j}^{2}).
\end{align}%
Also we write the second term $f_{2}:=(\prod_{i=1}^{M}r_{i}^{2})^{L}$.

Note that the whole expression is symmetric under permutation of $r_{i},r_{j}
$. Because of the axial symmetry of the Young tableau states, the
distribution must be uniform along the angular direction, otherwise it will
not correspond to axial symmetry. Therefore, if the maximum is unique, then
it must be at $r_{i}=r,\;\forall i$. This assumption needs to be verified,
and later we will verify this. Note also that the whole expression is
invariant under $\theta _{i}\rightarrow \theta _{i}+\Delta \theta $,
therefore we are free to choose $\theta _{1}=0$. Under the assumption that $%
r_{i}=r$, the expression can be simplified as
\begin{equation}
f=(1-r^{2})^{M}\prod_{1\leq i<j\leq M}(1-2r^{2}\cos (\theta _{i}-\theta
_{j})+r^{4})\times r^{2LM}.
\end{equation}

First consider the derivative with respect to $\theta _{l}$, $\frac{\partial
\log f}{\partial \theta _{l}}$, we get
\begin{equation}
\frac{\partial \log f}{\partial \theta _{l}}=\sum_{j=1}^{M}\frac{2r^{2}\sin
(\theta _{l}-\theta _{i})}{1-2r^{2}\cos (\theta _{l}-\theta _{i})+r^{4}}=0.
\end{equation}%
Solution to this equation is $\theta _{l}\equiv 0$ or $\theta _{l}=\frac{%
2\pi (l-1)}{M}$. For the first possibility, $\theta _{l}\equiv 0$. Because
in this case, $f_{1}=\prod_{i,j}(1-r_{i}r_{j})$, every factor is smaller
than 1, this cannot be a maximum. So we are left with the second solution $%
\theta _{l}=\frac{2\pi (l-1)}{M}$.

Then we consider determining $r$ that maximizes $f$. Consider the derivative
$\frac{\partial \log f}{\partial r}$, we have
\begin{equation}
\frac{\partial \log f}{\partial r}=\frac{2LM}{r}-\frac{2rM}{1-r^{2}}%
-4r\sum_{k=1}^{\infty }r^{2k-2}\sum_{1\leq l<j\leq M}\Re \left( e^{i\frac{%
2\pi k(j-l)}{M}}\right) .
\end{equation}%
Then if $k\equiv 0\mod M$, then $\sum_{1\leq l<j\leq M}\Re \left( e^{i\frac{%
2\pi k(j-l)}{M}}\right) =\frac{M(M-1)}{2}$. If $k\neq 0\mod M$, $\sum_{1\leq
l<j\leq M}\Re \left( e^{i\frac{2\pi k(j-l)}{M}}\right) =-\frac{M}{2}$.
Therefore we have
\begin{equation}
\frac{\partial \log f}{\partial r}=\frac{2LM}{r}-\frac{2M^{2}r^{2M-1}}{%
1-r^{2M}}=0.
\end{equation}%
Then $r_{0}=(\frac{L}{L+M})^{\frac{1}{2M}}$ solves this. Therefore a maximum
will be obtained at $\theta _{l}=\frac{2\pi (l-1)}{M}$, $r_{l}=r_{0}=(\frac{L%
}{L+M})^{\frac{1}{2M}}$. This configuration is very interesting, it means
that $x_{i}$ is uniformly distributed at the circle $r=r_{0}$.

Now evaluate $f$ at this point%
\begin{equation}
f=f_{1}f_{2}=\prod_{i,j}(1-r_{0}^{2}e^{i(\theta _{i}-\theta _{j})})\times
r_{0}^{2LM},
\end{equation}%
in which $f_{2}=r_{0}^{2LM}$. Observe that for any fixed $i\in \mathbb{Z}$,
when $j$ runs through $1,2,\dots ,M$, $\;i-j\mod M$ will also run through $%
1,2\dots ,M$. Therefore $e^{i\frac{2\pi (i-j)}{M}}$ will run through $e^{i%
\frac{2\pi }{M}},e^{i\frac{2\pi \times 2}{M}},\dots ,e^{i\frac{2\pi M}{M}}$.
Define a set $\Theta =\{\frac{2\pi l}{M}|l=1,2,\dots ,M\}$. Then $f_{1}$ can
be written as
\begin{equation}
f_{1}=(\prod_{\phi \in \Theta }(1-r_{0}^{2}e^{i\phi }))^{M}=\exp
(M\sum_{k=1}^{\infty }(-\sum_{\phi \in \Theta }\frac{r_{0}^{2k}e^{ik\phi }}{k%
})).
\end{equation}%
When $k\neq 0\mod M$, when $l$ runs through $1,2,\dots ,M$, $kl\mod M$ also
run through $1,2,\dots ,M$, therefore $\sum_{\phi \in \Theta }e^{ik\phi }=0$
for $k\neq 0\mod M$. When $k\equiv 0\mod M$, $e^{ik\phi }=1$, and $%
\sum_{\phi \in \Theta }e^{ik\phi }=M$. Therefore the above expression will
be
\begin{equation}
f_{1}=\exp (M\sum_{k=1}^{\infty }(-M\frac{r_{0}^{2kM}}{kM}%
))=(1-r_{0}^{2M})^{M}.
\end{equation}%
Then we maximize $f$ with respect to $r_{0}$.

And the final result is%
\begin{equation}
\sup f=(\frac{M}{L+M})^{M}(\frac{L}{L+M})^{L}.  \label{bound_01_}
\end{equation}

However we can't just conclude that the state with most overlap is obtained
at $x_{l}$ uniformly distributed around one circle in the complex plane. It
is possible that $x_{l}$ can uniformly distribute around two or more circles
but the angular positions coincide. But we can exclude this possibility by
doing computation. Let's compare the situations that $x_{l}$ uniformly
distributed around one circle and two circles. For $x_{l}$ distributed
around two circles write $M=2M^{\prime }$, $x_{l}=re^{i\frac{2\pi l}{%
M^{\prime }}}$ for $l=1,\dots ,M^{\prime }$ and $x_{l}=Re^{i\frac{2\pi l}{%
M^{\prime }}}$ for $l=M^{\prime }+1,\dots ,2M^{\prime }=M$. In this case, the norm squared of the overlap will be less than
\begin{align}
& (1-r^{2M^{\prime }})^{M^{\prime }}(1-R^{2M^{\prime }})^{M^{\prime
}}(r^{M^{\prime }}R^{M^{\prime }})^{2L}  \notag \\
& =(1-r^{2M^{\prime }})^{M^{\prime }}r^{2M^{\prime }L}(1-R^{2M^{\prime
}})^{M^{\prime }}R^{2M^{\prime }L}  \notag \\
& \leq ((\frac{M^{\prime }}{L+M^{\prime }})^{M^{\prime }}(\frac{L}{%
L+M^{\prime }})^{L})^{2}=(\frac{M}{2L+M})^{M}(\frac{2L}{2L+M})^{2L}. \label{appendix_B_01}
\end{align}%
And we have inequality
\begin{equation}
(\frac{M}{2L+M})^{M}(\frac{2L}{2L+M})^{2L}<(\frac{M}{L+M})^{M}(\frac{L}{L+M}%
)^{L}. \label{appendix_B_02}
\end{equation}%
Similar calculation can also be done for other situations like distribution around three or
more circles, and we have verified that all these situations have similar behavior as in the above case (\ref{appendix_B_01}) and (\ref{appendix_B_02}). The corresponding value is always smaller than the case for $x_l$ distribute uniformly around one circle. Hence, in conclusion, (\ref%
{bound_01_}) is a global maximum.
\end{proof}

This is the most strict upper bound, in the sense that there is a state that
can actually saturate this upper bound. We can consider the behavior of this
upper bound at both large $L$ and large $M$, and the above formula (\ref%
{bound_01_}) gives rise to
\begin{equation}
\sup f\leq f_{\text{bound}}=2^{-(L+M)},
\end{equation}%
where $f_{\text{bound}}$ denotes an upper bound of $f~$that is not
necessarily a supremum.

This upper bound has significance in the superposition of coherent states to
form a rectangular Young tableau state. The number of coherent states
participating in the superposition has consequently a lower bound. On one
hand, there are a large amount of states participating in the superposition.
On the other hand, the individual overlap is very small. This is reminiscent
to the observation in \cite{Berenstein:2017abm}.

\section{Chiral field and its variances on squeezed states}

\renewcommand{\theequation}{C.\arabic{equation}} \setcounter{equation}{0}

\label{appendix_variance}

Consider the field%
\begin{equation}
\hat{\phi}(\theta )=\sum_{m>0}a_{m}\exp (-im\theta )+a_{m}^{\dagger }\exp
(im\theta ).
\end{equation}%
Now we compute the variance $\langle $Squ$_{kk^{\prime }}\lvert :\hat{\phi}%
^{2}(\theta ):\rvert $Squ$_{kk^{\prime }}\rangle $, where $:\hat{\phi}%
^{2}(\theta ):$ is the normal ordering of $\hat{\phi}^{2}(\theta )$. First
\begin{equation}
:\hat{\phi}^{2}(\theta ):=\sum_{j,m>0}(a_{m}a_{j}\exp (-i(m+j)\theta
)+h.c.)+\sum_{j,m>0}(a_{m}^{\dagger }a_{j}\exp (i(m-j)\theta )+h.c.).
\end{equation}

This gives
\begin{align}
& \langle \text{Squ}_{kk^{\prime }}\lvert :\hat{\phi}^{2}(\theta ):\rvert
\text{Squ}_{kk^{\prime }}\rangle  \notag \\
& =2\Re \left( 2\langle \text{Squ}_{kk^{\prime }}\lvert a_{k}a_{k^{\prime
}}\exp (-i(k+k^{\prime })\theta )\rvert \text{Squ}_{kk^{\prime }}\rangle
+\sum_{j,m>0}\langle \text{Squ}_{kk^{\prime }}\lvert a_{m}^{\dagger
}a_{j}\exp (i(m-j)\theta )\rvert \text{Squ}_{kk^{\prime }}\rangle \right) .
\notag \\
&
\end{align}

While the first term can be computed to be
\begin{align}
& 4\Re \left( \langle \text{Squ}_{kk^{\prime }}\lvert a_{k}a_{k^{\prime
}}\exp (-i(k+k^{\prime })\theta )\rvert \text{Squ}_{kk^{\prime }}\rangle
\right)  \notag \\
& =4(1-q^{2})\Re \left( \exp (-i(k+k^{\prime })\theta )\sum_{l}(\frac{q}{%
\sqrt{kk^{\prime }}})^{l}\frac{1}{l!}(\frac{q}{\sqrt{kk^{\prime }}})^{l+1}%
\frac{kk^{\prime }(l+1)^{2}}{(l+1)!}(kk^{\prime })^{l}l!^{2}\right)  \notag
\\
& =4(1-q^{2})\cos ((k+k^{\prime })\theta )\left( \sum_{l}q^{2l+1}\sqrt{%
kk^{\prime }}(l+1)\right)  \notag \\
& =\frac{4q}{(1-q^{2})}\sqrt{kk^{\prime }}\cos ((k+k^{\prime })\theta ).
\end{align}

The second term is
\begin{align}
& 2\Re \left( \sum_{j,m>0}\langle \text{Squ}_{kk^{\prime }}\lvert
a_{m}^{\dagger }a_{j}\exp (i(m-j)\theta )\rvert \text{Squ}_{kk^{\prime
}}\rangle \right) =2\Re \left( \sum_{m>0}\langle \text{Squ}_{kk^{\prime
}}\lvert a_{m}^{\dagger }a_{m}\rvert \text{Squ}_{kk^{\prime }}\rangle \right)
\notag \\
& =2\Re \left( \langle \text{Squ}_{kk^{\prime }}\lvert a_{k}^{\dagger
}a_{k}\rvert \text{Squ}_{kk^{\prime }}\rangle \right) +2\Re \left( \langle
\text{Squ}_{kk^{\prime }}\lvert a_{k^{\prime }}^{\dagger }a_{k^{\prime
}}\rvert \text{Squ}_{kk^{\prime }}\rangle \right)  \notag \\
& =2(1-q^{2})\left( \sum_{l}q^{2l}(k+k^{\prime })l\right)  \notag \\
& =\frac{2q^{2}}{(1-q^{2})}(k+k^{\prime }).
\end{align}

Thus the variance is
\begin{equation}
\langle \text{Squ}_{kk^{\prime }}\lvert :\hat{\phi}^{2}(\theta ):\rvert
\text{Squ}_{kk^{\prime }}\rangle =\frac{2q}{(1-q^{2})}(q(k+k^{\prime })+2%
\sqrt{kk^{\prime }}\cos ((k+k^{\prime })\theta )).
\label{variance_squeezed_01}
\end{equation}

We can also calculate higher order variances, $\langle \text{Squ}%
_{kk^{\prime }}\lvert :\hat{\phi}^{p}(\theta ):\rvert \text{Squ}_{kk^{\prime
}}\rangle $, using a generating function technique. Define $M(t)=\langle
\text{Squ}_{kk^{\prime }}\lvert :\exp (t\hat{\phi}(\theta )):\rvert \text{Squ%
}_{kk^{\prime }}\rangle $. We make use of the fact that $:\exp (t\hat{\phi}%
(\theta )):=\prod_{m>0}\exp (ta_{m}^{\dagger }e^{im\theta })\exp
(ta_{m}e^{-im\theta })$. We calculate that
\begin{equation}
M(t)=\exp \left[ \frac{q}{1-q^{2}}t^{2}(q(k+k^{\prime })+2\sqrt{kk^{\prime }}%
\cos (k+k^{\prime })\theta )\right] .  \label{generating_02}
\end{equation}%
Using the relation
\begin{equation}
M(t)=\sum_{p}\frac{t^{p}}{p!}\langle \text{Squ}_{kk^{\prime }}\lvert :\hat{%
\phi}^{p}(\theta ):\rvert \text{Squ}_{kk^{\prime }}\rangle ,
\end{equation}%
and expanding (\ref{variance_squeezed_01}), we have that
\begin{align}
& \langle \text{Squ}_{kk^{\prime }}\lvert :{\hat{\phi}}^{2l-1}(\theta
):\rvert \text{Squ}_{kk^{\prime }}\rangle =0,  \notag \\
& \langle \text{Squ}_{kk^{\prime }}\lvert :\hat{\phi}^{2l}(\theta ):\rvert
\text{Squ}_{kk^{\prime }}\rangle =\frac{(2l)!}{l!}\left( \frac{q}{1-q^{2}}%
(q(k+k^{\prime })+2\sqrt{kk^{\prime }}\cos (k+k^{\prime })\theta )\right)
^{l}.
\end{align}%
For $l=1$, this gives the above variance formula (\ref{variance_squeezed_01}%
).


\end{document}